\documentclass[a4paper, 11pt]{article}
\usepackage[hmargin=1in,vmargin=1in]{geometry}
\usepackage{style}
\usepackage[noend, ruled]{algorithm2e}

\newtheorem{theorem}{Theorem}[section]
\newtheorem{lemma}[theorem]{Lemma}

\newtheorem{corollary}[theorem]{Corollary}

\newtheorem{definition}[theorem]{Definition}

\hypersetup{
    colorlinks=true,
    linkcolor=red!70!black,
    linktoc=all,
    citecolor=blue
}

\newcommand{\qedhere}{}

\author{
    \entry{Sujoy Bhore}{iitb}{Department of Computer Science \& Engineering, Indian Institute of Technology Bombay, Mumbai, India.  This work is supported in part by ANRF ARG-MATRICS, Grant 002465.}{sujoy@cse.iitb.ac.in} &
    \entry{Karl Bringmann}{uds}{Saarland University and Max-Planck-Institute for Informatics, Saarland Informatics Campus, Saarbr\"ucken, Germany. This work is part of the project TIPEA that has received funding from the European Research Council (ERC) under the European Unions Horizon 2020 research and innovation programme (grant agreement No. 850979).}{bringmann@cs.uni-saarland.de} \\
    \entry{Timothy M. Chan}{uiuc}{Department of Computer Science, University of Illinois at Urbana-Champaign, USA\@.  This work is supported in part by NSF Grant CCF-2224271.}{tmc@illinois.edu} &
    \entry{Yanheng Wang}{uds}{}{yanhwang@cs.uni-saarland.de}
}

\mathcode`l="8000
\begingroup
\makeatletter
\lccode`\~=`\l
\DeclareMathSymbol{\lsb@l}{\mathalpha}{letters}{`l}
\lowercase{\gdef~{\ifnum\the\mathgroup=\m@ne \ell \else \lsb@l \fi}}%
\endgroup

\newcommand{\size}{\textsc{size}}
\newcommand{\sample}{\textsc{sample}}
\newcommand{\contains}{\textsc{contains}}
\newcommand{\Pois}{\mathrm{Pois}}
\newcommand{\level}{\mathrm{level}}
\newcommand{\supp}{\mathrm{supp}}
\newcommand{\Alg}{\mathcal{A}}
\newcommand{\F}{\mathbb{F}}
\newcommand{\eps}{\varepsilon}
\renewcommand{\Pr}{\mathbb{P}}
\newcommand{\class}[1]{\mathcal{#1}}
\newcommand{\discrete}[1]{#1^\diamond}

\DontPrintSemicolon
\SetFuncSty{textsc}
\SetArgSty{}
\SetKwProg{Function}{function}{}{}
\SetKwInput{Parameter}{parameters}
\SetKwInput{Public}{public variables}

\SetKwFunction{initialize}{initialize}
\SetKwFunction{ins}{insert}
\SetKwFunction{del}{delete}
\SetKwFunction{est}{estimate}
\SetKwFunction{ILP}{list}
\SetKwFunction{refresh}{refresh}
\SetKwFunction{KLM}{KLM}
\SetKwFunction{resample}{erase-resample}
\SetKwFunction{filter}{filter}
\SetKwFunction{update}{update}
\SetKwFunction{support}{support}

\title{Dynamic and Streaming Algorithms\\for Union Volume Estimation}

\begin{document}

\maketitle

\begin{abstract}
The \emph{union volume estimation} problem asks to $(1\pm\eps)$-approximate the volume of the union of $n$ given objects $X_1,\ldots,X_n \subset \R^d$. In their seminal work in 1989, Karp, Luby, and Madras solved this problem in time $O(n/\eps^2)$ in an oracle model where each object $X_i$ can be accessed via three types of queries: obtain the volume of $X_i$, sample a random point from $X_i$, and test whether $X_i$ contains a given point $x$. This running time was recently shown to be optimal~[Bringmann, Larsen, Nusser, Rotenberg, and Wang, SoCG'25]. In another line of work, Meel, Vinodchandran, and Chakraborty [PODS'21] designed algorithms that read the objects in one pass using polylogarithmic time per object and polylogarithmic space; this can be phrased as a dynamic algorithm supporting insertions of objects for union volume estimation in the oracle model.

In this paper, we study algorithms for union volume estimation in the oracle model that support both \emph{insertions} and \emph{deletions} of objects. We obtain the following results:
\begin{enumerate}
    \item an algorithm supporting insertions and deletions in polylogarithmic update and query time and linear space (this is the first such dynamic algorithm, even for 2D triangles);
    \item an algorithm supporting insertions and suffix queries (which generalizes the sliding window setting) in polylogarithmic update and query time and space;
    \item an algorithm supporting insertions and deletions of convex bodies of constant dimension in polylogarithmic update and query time and space.
\end{enumerate}
\end{abstract}
\newpage

\setcounter{tocdepth}{2}
	\tableofcontents

\newpage	
\setcounter{page}{1}

\section{Introduction}
In the \emph{union volume estimation} problem, we are given \emph{objects} $X_1,\ldots,X_n$, which are assumed to be measurable sets in $\R^d$, and the task is to $(1\pm\eps)$-approximate the volume of $\bigcup_{i=1}^n X_i$.%
\footnote{A discrete variant of the problem, estimating the \emph{cardinality} of the union of sets $X_1,\ldots,X_n \subset \Z^d$, can be easily reduced to the continuous version, by replacing each grid point with a unit grid cell.}
This is a fundamental problem in computer science, with connections to several areas such as network reliability~\cite{Karger01, KarpLM89}, DNF counting~\cite{KarpLM89}, general model counting~\cite{pavan2023model, chakraborty2013scalable}, probabilistic databases~\cite{kimelfeld2008query, DalviS07, re2006efficient}, and probabilistic query evaluation on databases~\cite{carmeli2020answering, carmeli2023tractable}.
Perhaps the most popular special case in computational geometry is 
\emph{Klee's measure problem}~\cite{Klee77}, where each object is a $d$-dimensional axis-aligned box of the form $[a_1, b_1] \times \cdots \times [a_d, b_d]$. There has been extensive work on Klee's measure problem (either exact or approximate) in constant dimensions~\cite{OY91,Chan10,Chan13,Kunnemann22,GorbachevK23,BLNRW25}.

To capture a wide variety of input objects with a unified framework, Karp and Luby~\cite{KarpL85} introduced an \emph{oracle model}. An algorithm in this model can access each object $X_i$ exclusively via three types of queries: (i)~$X_i.\size()$ gives the volume of $X_i$, (ii)~$X_i.\sample()$ samples a point uniformly at random from the set $X_i$, and (iii)~$X_i.\contains(x)$ tests whether a given point $x$ is contained in $X_i$. The time complexity of the algorithm is measured by the number of queries plus the number of additional steps; that is, each query is assumed to take unit time. In a seminal result, Karp, Luby, and Madras~\cite{KarpLM89} gave an algorithm in the oracle model that solves union volume estimation in time $O(n/\eps^2)$ with constant success probability, where $\eps$ is the approximation precision. Black-box usage of this result also yields the best known algorithms for several applications, including DNF counting. (For Klee's measure problem, it implies an $O(nd / \eps^2)$-time approximation algorithm~\cite{BringmannF10}; in small dimensions, this was improved to $O((n+ 1/\eps^2)\log^{O(d)}n)$ time at last year's SoCG \cite{BLNRW25}.) Only recently has a matching lower bound been established: $\Omega(n/\eps^2)$ queries are necessary to solve union volume estimation in the oracle model~\cite{BLNRW25}.

With the static setting now fully understood, attention naturally shifts to the dynamic setting, where the input evolves over time through insertions and deletions of objects. The study of \emph{dynamic union volume estimation} has a twofold motivation. On one hand, it captures modern data-driven applications where datasets are inherently mutable, and volume estimates must be maintained continuously rather than recomputed from scratch after each update. On the other hand, the dynamic setting presents a compelling theoretical challenge, as it is far from obvious whether the tools developed for the static setting can be dynamized efficiently.
Indeed, dynamic geometric problems have repeatedly proven to be substantially more difficult than their static counterparts. An extensive body of work in recent years deals with dynamic geometric problems such as set cover \cite{AgarwalCSXX22, ChanHSX25}, independent set \cite{Henzinger0W20, BhoreNTW24, BhoreC25}, planar point location \cite{Nekrich21}, connectivity \cite{ChanPR11, ChanH24a}, and nearest neighbor search \cite{AgarwalAS18}, yet in this broader context dynamic union volume estimation remains largely unexplored. The following question has not been resolved even for very simple classes of objects such as triangles in the plane:

\begin{quote}
\begin{center}
   \emph{Q1: Can we solve dynamic union volume estimation in the oracle model under insertions and deletions of objects in $\poly \left( \log n, \frac{1}{\eps} \right)$ update and query time?}\footnote{Low-dimensional fat objects can be easily handled using quadtrees; see Section~\ref{sec:fat} for details. The problem becomes much more challenging for non-fat objects.}
\end{center}
\end{quote}

An even more demanding paradigm is the (one-pass) streaming setting, where the algorithm has to not only handle updates, but use space that is typically required to be polylogarithmic in the input size. Earlier work has explored streaming algorithms under strong restrictions on the input objects. For example, Tirthapura and Woodruff~\cite{TW12} designed a dynamic streaming algorithm for Klee’s measure problem on axis-aligned boxes. However, the problem remains unaddressed for more general classes of objects:

\begin{quote}
\begin{center}
   \emph{Q2: Can we solve dynamic union volume estimation in the oracle model under insertions and deletions of objects in $\poly \left( \log n, \frac{1}{\eps} \right)$ space as well as update and query time?}
\end{center}
\end{quote}

In the restricted \emph{insertion-only} setting, Meel, Vinodchandran, and Chakraborty~\cite{MVC21} gave a positive answer to this question, using $O\left(\frac{1}{\eps^2} \log^2(n) (\log \log(n) + \log(\frac{1}{\eps})) \right)$ time per update and query, and $O\left( \frac{1}{\eps^2} \log(n) \right)$ space in the oracle model.%
\footnote{In this paper, we assume that points from the input objects and $O(\log n)$-bit integers fit into a machine word, and space complexity counts the number of machine words. An arithmetic operation on machine words is assumed to take unit time. Both are natural assumptions for a continuous universe. Some previous work such as \cite{MVC21} focuses on sets from a finite universe and studies bit complexity; when stating previous work in the introduction, we have translated their results into our setting. In the finite universe $\Omega$, different trade-offs between $\log(n)$ and $\log |\Omega|$ were obtained in~\cite{MCV22,NandiVGMP024}.}
However, their algorithm cannot handle deletions.

\section{Our Contributions}

\subsection{Dynamic algorithm with insertions and deletions}
We answer our driving question (Q1) positively, not only for 2D triangles, 3D simplices, and other common types of objects encountered in computational geometry, but in the full generality of the oracle model.  Our new dynamic algorithm can handle arbitrary sequences of both insertions and deletions.

More precisely, we consider an initially empty multiset of objects $\class{X}$ under the updates \ins{$X$}, which adds object $X$ to $\class{X}$, and \del{$X$}, which removes object $X$ from $\class{X}$. Upon query \est{}, the algorithm outputs a $(1\pm\eps)$-approximation to the volume $\vol\left(\bigcup_{X \in \class{X}} X\right)$.

During insertion/deletion the algorithm is given oracle access to the inserted/deleted object. We also need the ability to \emph{store} objects (or their oracles), in case of future lookup of the previously inserted objects. We make the natural assumption that (i) the argument of \ins{$X$} and \del{$X$} is a \emph{pointer} to the oracle of $X$, (ii) we can store a pointer in a machine word, and (iii) knowing the pointer allows us to access the oracle to which the pointer is pointing. In this interpretation, $\class{X}$ is treated as a multiset of pointers: the same pointer can be inserted multiple times, and upon deletion, the multiplicity of the pointer is decreased by one. Naturally, the multiplicity of any pointer must be non-negative throughout the process (in particular, objects being deleted must have previously been inserted).

\begin{theorem}[Section~\ref{sec:dynamic}]
\label{thm:dynamicalgo}%
There is a dynamic algorithm in the oracle model that maintains a multiset $\class{X}$ of objects under updates \ins{$X$} and \del{$X$}. Upon query \est{}, it outputs a $(1\pm\eps)$-approximation to $\vol\left(\bigcup_{X \in \class{X}} X\right)$ with high probability. The amortized expected update time is $O(\log^5(n/\eps)/\eps^2)$, the query time is $O(\log^4(n)/\eps^2)$, and the space complexity is $O(n + \log^4(n)/\eps^2)$.
\end{theorem}

Note that since the algorithm operates in the general oracle model, it works well even for \emph{high-dimensional} objects (e.g., for simplices in $\R^d$, with polynomial dependence on $d$).

\subsection{Streaming algorithm with insertions and suffix queries}

In our remaining results, we turn to streaming algorithms and our driving question (Q2) and show that polylogarithmic time \emph{and} space can be achieved in certain settings. 
We first consider the \emph{sliding window} setting~\cite{datar2002maintaining}, which supports insertions of objects and querying over the last $w$ objects for a fixed window size $w$;
thus, objects are deleted in the same order as they are inserted. Numerous results and techniques have been developed in the sliding window setting (see~\cite{muthukrishnan2005data,babcock2002models,Braverman16} for an overview and~\cite{woodruff2022tight,JayaramWZ22,FengSW25} for recent work), but no such result was known for union volume estimation.

We obtain the first efficient sliding-window algorithm for union volume estimation. Again, it
works in the general oracle model! Moreover, the result holds in an extension of the sliding-window setting, where at each time $t = 1, \dots, n$ an object $X_t$ is inserted (but cannot be deleted), and a query \est{$s$} may return a $(1\pm\eps)$-approximation to the volume of the \emph{suffix} $\vol(X_s \cup \cdots \cup X_t)$ with high probability for \emph{any} number $s$ the user specifies. 

\begin{theorem}[Section~\ref{sec:window}]
\label{thm:windowalgo}%
There is an algorithm that implicitly maintains a sequence of objects $X_1, \dots, X_t$ of aspect ratio $M$. The algorithm supports update \ins{$X$}, which appends an object $X$ (as $X_{t+1}$) to the sequence, and query \est{$s$}, which outputs a $(1\pm\eps)$-approximation to $\vol(X_s \cup \cdots \cup X_t)$ with high probability. The update time is $O(\log^2(n/\eps)/\eps^2)$, the query time is $O(\log(n)/\eps^2)$, and the space complexity is $O(\log(M) \log(n)/\eps^2)$.
\end{theorem}

Our algorithm not only generalizes previous insertion-only results~\cite{MVC21}, but is also remarkably simple. Note that the space bound depends on the aspect ratio $M$, defined as the ratio of the largest volume of any object ever inserted and the smallest volume of any object ever inserted. In many applications, $M$ is a polynomial in $n$ and thus $\log(M) = O(\log n)$. Furthermore, geometric algorithms in the sliding window setting~\cite{FeigenbaumKZ04,ChanS06} usually have logarithmic dependencies on parameters such as the aspect ratio or spread.

\subsection{Streaming algorithm for convex bodies with insertions and deletions}
Lastly, we consider streaming algorithms that support \emph{both} insertions and deletions.
We manage to develop an efficient algorithm for a reasonably wide class of geometric objects---namely, convex bodies in any constant dimension:

\begin{theorem}[Section~\ref{sec:convex}]
\label{thm:convexalgo}
Given $R \ge r > 0$, there is an algorithm that implicitly maintains a multiset $\class{X}$ of convex bodies in $[0,R]^d$, where each convex body contains a ball of radius $r$, under updates \ins{$X$} and \del{$X$}. Upon query \est{}, it outputs a $(1\pm\eps)$-approximation to $\vol\left(\bigcup_{X \in \class{X}} X\right)$ with high probability. The expected update and query time and space are $O(\polylog(\frac{nR}{\eps r}) / \eps^2)$ for constant $d$.
\end{theorem}

The bounds depend polylogarithmically on $R/r$, which is common in volume estimation algorithms for a single convex object~\cite{DyerFK91} as well as geometric streaming algorithms~\cite{Indyk04,FrahlingIS08}.

Previously, Tirthapura and Woodruff obtained an efficient dynamic streaming algorithm only for axis-aligned boxes~\cite{TW12} (with $\poly(\log n, \log \Delta, \frac{1}{\eps}, d)$ update time and space in a discrete universe $\{1,\ldots,\Delta\}^d$).
To emphasize the generality of our algorithm, note that many non-convex objects can be decomposed into a small number of convex objects, e.g., non-convex polyhedra of constant combinatorial complexity in constant dimensions can be decomposed into a constant number of simplices, to which our algorithm can be  applied.

\section{Technical Overview}

\subsection{Dynamic algorithm with insertions and deletions}
Karp, Luby and Madras' algorithm~\cite{KarpLM89} and subsequent work are based on the observation that we can estimate the volume of a union $U$ by maintaining a finite random subset $S^{(l)} \subset U$ that selects every point independently with probability density $1/2^l$. It turns out that the ``right'' choice of $l$ is one such that $|S^{(l)}| \approx \log(n)/\eps^2$; in this case $2^l|S^{(l)}|$ provides a good estimate of the volume of $U$.

Meel, Vinodchandran, and Chakraborty~\cite{MVC21} described a particularly simple way to maintain such a random set $S^{(l)}$ for insertion-only updates: when a new object $X$ is inserted, we remove from $S^{(l)}$ all points that are contained in $X$, then add to $S^{(l)}$ freshly sampled points in $X$ with density $1/2^l$. This step is inexpensive for the right choice of $l$ since $|S^{(l)}|$ is logarithmic. If $|S^{(l)}|$ becomes too big, we increment $l$ (thereby halving the sampling rate and also $|S^{(l)}|$ in expectation) and move on.

Difficulty arises in the presence of deletions: the volume is no longer monotonically increasing, and the right choice of $l$ may increase or decrease over time. Our first observation is that the \emph{deletion-only} case can be handled by a variant of Meel, Vinodchandran, and Chakraborty's insertion-only algorithm. Roughly, we can maintain $S^{(l)}$ for the right $l$ by keeping a counter on the number of objects containing each sample point (recall that $|S^{(l)}|$ is logarithmic for the right $l$). When an object $X$ is deleted, we decrement the counters of all points in $S^{(l)} \cap X$, and when a counter drops to 0, the point is removed from $S^{(l)}$. When the number of surviving sample points drops below logarithmic, we decrease $l$, rebuild $S^{(l)}$ by iterating over all objects that has not been deleted (recall that can store pointers to all objects), and reinitialize the counters.

A popular technique in dynamic computational geometry by Bentley and Saxe~\cite{BS80}, called the \emph{logarithmic method}, transforms deletion-only data structures into fully dynamic ones. It works by decomposing a fully dynamic set into logarithmically many deletion-only subsets, increasing the update and query time by typically just a logarithmic factor. The logarithmic method applies to so-called \emph{decomposable} problems that satisfy the following property:
if the input is decomposed into subsets, then the answer to the original problem can be quickly obtained from the answers to the individual subsets. Unfortunately, union volume estimation is \emph{not} a decomposable problem: We cannot determine the volume of the union of $\class{X}_0\cup \cdots \cup \class{X}_{\log n}$ just from the volume of the union of each $\class{X}_i$.

Nevertheless, we show that a novel variant of the logarithmic method can still work. Recall that the deletion-only data structure already maintains a sample for each $\class{X}_i$. For each sample point we keep track of more counters---specifically, how many sets from the other $\class{X}_i$'s contain the point. These counters provide enough information to avoid double counting and help us estimate the union volume accurately.
One issue though is that many counters need to be adjusted when one of the deletion-only data structures changes $l$ and recomputes $S^{(l)}$. Fortunately, we manage to bound the total number of such changes, so the amortized cost for maintaining all counters remains polylogarithmic. See Section~\ref{sec:dynamic} for details.

\subsection{Streaming algorithm with insertions and suffix queries}

The dynamic algorithm with arbitrary insertions and deletions needs linear space because the data structure rebuilds on the remaining objects when the ``right'' $l$ decreases. However, in the sliding window setting and more generally the setting of suffix queries, the remaining objects are contiguous in time and need not be stored explicitly. We obtain a streaming algorithm in this case by a simple yet non-trivial generalization of Meel, Vinodchandran, and Chakraborty's algorithm. Specifically, we keep the timestamp when each point in $S^{(l)}$ is sampled. For each $l$ we keep track of the largest timestamp $s^{(l)}$ such that discarding points older than $s^{(l)}$ from $S^{(l)}$ makes $|S^{(l)}|$ logarithmically bounded. This timestamp $s^{(l)}$ is monotonically increasing over time for each fixed $l$. To estimate the volume of the union of $X_s \cup \cdots \cup X_t$ for a given $s$ at time $t$, we identify the smallest $l$ for which $s^{(l)} \leq s$, and output $2^l|S^{(l)}\cap (X_s \cup \ldots \cup X_t)|$. See Section~\ref{sec:window} for details.

\subsection{Streaming algorithm for convex bodies with insertions and deletions}
Our streaming algorithm for low-dimensional convex objects for arbitrary insertions and deletions is based on another simple sampling strategy. First, it is not difficult to discretize the space $\R^d$ to a grid $[\Delta]^d$ for some large $\Delta$, so that the volume of the union $U$ is close to the number of grid points in $U \cap [\Delta]^d$. The idea now is to select a random subset $R^{(l)}$ of the grid $[\Delta]^d$ in the beginning, and maintain the intersection $S^{(l)} := U \cap R^{(l)}$ during the updates. Using a pairwise independent family of hash functions, $R^{(l)}$ can be represented implicitly with logarithmically many bits.

When we insert or delete an object $X$, we just need to insert or delete $X \cap R^{(l)}$ to or from $S^{(l)}$ (viewed as a multiset of points).
An issue is that the space usage may be large for a $S^{(l)}$ with a ``wrong'' sampling rate, and we do not know which $l$ is the right choice beforehand. Furthermore, a wrong choice may become right at a later time.
This issue can be resolved by representing $S^{(l)}$ by data structures for sparse vectors~\cite{FrahlingIS08,JowhariST11}, which support adding numbers to a given coordinate, as well as recovering the vector when its norm is below a threshold. In this way, we essentially just pay for the space at the right sampling rate.

A question remains: how to enumerate the points of $X \cap R^{(l)}$ for a given object $X$ and for our implicitly represented set $R^{(l)}$? For the case where the objects are rotated (non-axis-aligned) boxes in constant dimension, and for a particular pairwise independent family of hash functions, we observe that this subproblem directly reduces to integer linear programming (ILP) with a constant number of constraints and variables, and therefore can be solved efficiently by well-known algorithms~\cite{Lenstra83}.
For more general convex objects in constant dimension, we bound each object $X$ by a small rotated box $B \supseteq X$, use an ILP to enumerate the points of $B \cap R^{(l)}$ and discard those not in $X$. By ``small'' we mean that the volume of $B$ is comparable to that of $X$, so in particular the number of discarded points is small. See Section~\ref{sec:convex} for the details.

Our approach is related to Tirthapura and Woodruff's algorithm~\cite{TW12} for axis-aligned boxes. Their idea was to directly adapt a streaming algorithm for approximately counting distinct elements to estimate $|U\cap [\Delta]^d|$: 
to insert or delete an object $X$, we cannot afford to literally insert or delete all points in $|X\cap [\Delta]^d|$, but they showed how to efficiently simulate the effect for a specific distinct elements algorithm in the case of axis-aligned boxes.
Our approach is more modular in that we apply sparse recovery only as a black box.
Our main contribution however is the realization, via solving an ILP, that this type of approach works far more generally than for axis-aligned boxes, namely, for general convex objects.

\section{Preliminaries}
\label{sec:preliminaries}

\subsection{Technical setup}
We assume that all relevant parameters (the precision $\eps$, an upper bound $n$ on the number of operations, etc.) are given to the algorithms at initialization. Knowing $n$ is in line with the usual assumptions in streaming algorithms; it can easily be avoided but this complicates notation. The phrase ``with high probability'' means that the probability is at least $1 - n^{-\Omega(1)}$.

We denote $[n] := \set{1, 2 \dots, n}$. The notation $a \in b \pm c$ means that $b-c \leq a \leq b + c$. Similarly, $a \in (1 \pm \eps) b$ means that $(1-\eps)b \leq a \leq (1 + \eps)b$, and we say that $a$ is a $(1\pm\eps)$-approximation to $b$.

\subsection{Reduction of volume range}
The following lemma reduces the volumes of objects to a narrow range $[(3n/\eps)^2, (3n/\eps)^4]$, without affecting efficiency much.
\begin{lemma}
    \label{lem:reduce-aspect-ratio}
    Suppose that $(1 \pm \eps)$ union volume estimation under insertions and deletions (resp. suffix queries) can be solved in $S(n,\eps)$ space and $T(n,\eps)$ update time when objects have volumes in $[(3n/\eps)^2, (3n/\eps)^4]$. Then the problem without volume restriction can be solved in $S(n,\eps/3) \cdot O\left(\frac{\log(M)}{\log(n/\eps)} + 1\right)$ space and $T(n,\eps/3) + O(\log n)$ update time. Moreover, if $S(n,\eps)/n$ is non-decreasing in $n$, then the space bound can be strengthened to $O(S(n,\eps/3))$.
\end{lemma}

\begin{proof}
    Write $m := 3n/\eps$. Let $\Alg$ be a union volume estimation algorithm using $S(n,\eps)$ space and $T(n,\eps)$ update time when objects have volumes in $[(3n/\eps)^2, (3n/\eps)^4]$. We construct another algorithm that handles general volumes.
    
    Let class $\class{C}_l$ collect all input objects whose volume is in $(m^{l-1}, m^{l+1}]$. Non-empty classes are said to be \emph{active}. Note that each object appears in exactly two classes, and there are $O(\frac{\log M}{\log m} + 1)$ active classes. We organize the active classes in a binary search tree ordered by indices $l$. However, we do not store $\class{C}_l$ explicitly but only a counter $|\class{C}_l|$. Moreover, we run a copy $\Alg_l$ of $\Alg$ with parameter $\eps/3$ on each active class $\class{C}_l$, after scaling the objects in it by a factor of $m^{3-l}$. The objects inside the class now have volumes in $(m^2, m^4]$, so algorithm $\Alg$ can solve union volume estimation in space $S(n,\eps/3)$ and update time $T(n,\eps/3)$.
    
    When an object $X$ is inserted, we determine from $X.\size$ the two classes $\class{C}_l, \class{C}_{l'}$ containing $X$. If $C_l$ or $C_{l'}$ was empty before the insertion, then we add it to the binary search tree. We increment the counters $|\class{C}_l|$ and $|\class{C}_{l'}|$. Finally, we insert $X$ (after scaling) into $\Alg_l$ and $\Alg_{l'}$. Deletion is handled similarly, and if some counter drops to zero we remove the class from the tree. To estimate the volume of the union, we locate the second largest $l$ in the tree, and output the estimate of $\Alg_l$ scaled by a factor of $m^{l-3}$.

    It is clear that the space and update time of this new algorithm is as stated. To see correctness, we first observe that the highest class is a subset of the second highest. Indeed, every object $X$ is in exactly two classes; if it is in the highest then it needs to be in the second highest, too.
    Next we write $z$ for the correct output, $z_l$ for the volume of union of the second highest class, and $z_{<l}$ for the volume of union of lower classes. Note that there exists an object of volume at least $m^l$, because $l$ is not the highest class. In particular, $z_l \geq m^l$. On the other hand, every object contained in some lower class but not in the second highest class has volume at most $m^{l-1}$ by definition, so $z_{<l} \leq n \cdot m^{l-1} \leq \eps z_l/3$. Therefore,
    \[ z_l \leq z \leq z_l + z_{<l} \leq \left( 1+\frac{\eps}{3} \right) z_l. \]
    It follows that any $(1\pm\frac{\eps}{3})$-approximation to $z_l$ is a $(1\pm\frac{\eps}{3})^2 \subseteq(1\pm\eps)$-approximation to $z$.

    For the ``moreover'' part, let us write $n_l := |\class{C}_l|$. The space of the same algorithm can be bounded up to constant factors by
    \[
        \sum_l S \left( n_l, \frac{\eps}{3} \right)
        = \sum_l n_l \cdot \frac{S(n_l, \frac{\eps}{3})}{n_l}
        \leq \sum_l n_l \cdot \frac{S(n, \frac{\eps}{3})}{n}
        \leq 2 n \cdot \frac{S(n, \frac{\eps}{3})}{n}
        = 2 S\left( n, \frac{\eps}{3} \right),
    \]
    where the second step uses the assumption that $S(n,\eps)/n$ is non-decreasing, and the third step uses the fact that each object appears in two classes.
\end{proof}

\subsection{Volume estimation via sampling}
As in all previous work, our algorithms exploit the close relation between volume estimation and sampling. The notion of $l$-samples is central:

\begin{definition}[$l$-sample]
    Let $U$ be a measurable set and $l \in \Z$. A random finite subset $S \subset U$ is an \emph{$l$-sample} of $U$, denoted $S \sim \Pois(U,l)$, if $|S \cap V| \sim \Pois(\vol(V)/2^l)$ for every measurable subset $V \subseteq U$.
\end{definition}

It follows from this definition that: (i) if $S \sim \Pois(U,l)$, then $S \cap V \sim \Pois(V,l)$ for any measurable subset $V \subseteq U$; (ii) if $S \sim \Pois(U,l)$ and $T \sim \Pois(V,l)$ are independent and $U,V$ are disjoint, then $S \cup T \sim \Pois(U \cup V, l)$.

Note that an $l$-sample $S$ of $U$ has expected size $\vol(U)/2^l$.
If $l$ is small, then intuitively $2^l |S|$ should concentrate around $\vol(U)$. The threshold of being small is given in Definition~\ref{def:level}, and its justification in Lemma~\ref{lem:l-sample}.
\begin{definition}[level]
    \label{def:level}
    For a measurable set $U$, we define $\level(U)$ as the unique integer such that $\frac{32 \log n}{\eps^2} \leq \frac{\vol(U)}{2^{\level(U)}} < \frac{64 \log n}{\eps^2}$.
\end{definition}

\begin{lemma}
    \label{lem:l-sample}
    If $S$ is an $l$-sample of $U$ for $l \leq \level(U)$, then $2^l |S| \in (1\pm\eps) \vol(U)$ with probability at least $1 - O(n^{-8})$.
\end{lemma}

\begin{proof}
    Note that $|S| \sim \Pois(\vol(U)/2^l)$. By concentration of Poisson variables~\cite{Canonne19},
    \[
        \Pr \left( \left| 2^l |S| - \vol(U) \right| \geq \eps \vol(U) \right)
        \leq 2\exp\left( \frac{-\eps^2 \vol(U)}{4 \cdot 2^l} \right)
        \leq 2\mathrm{e}^{-8 \log n} = O(n^{-8}).\qedhere
    \]
    \aftermath
\end{proof}

Let $U$ denote the union we are interested in. Each of our algorithms morally maintains an $l$-sample $S^{(l)}$ of $U$ for every $l \in \Z$ throughout the updates; upon query, it tries to identify some $L \leq \level(U)$ and outputs $2^L |S^{(L)}|$. In the analysis, we morally apply Lemma~\ref{lem:l-sample} to conclude that $2^L |S^{(L)}|$ is a good estimate of $\vol(U)$. However, there are two pitfalls. The first pitfall is that the identification of $L$ actually depends on the outcomes of the random sets $S^{(l)}$, so $L$ is a random variable. Conditioned on~$L$, it is not true that $S^{(L)} \sim \Pois(U, L)$, and Lemma~\ref{lem:l-sample} does not apply directly.

The second pitfall is that we cannot afford to maintain all sets $S^{(l)}$ at once. Our algorithms only explicitly store those $S^{(l)}$ with $l \approx \level(U)$; since such $S^{(l)}$ has an expected size roughly $\vol(U)/2^{\level(U)} = \Theta(\log(n)/\eps^2)$ it can be stored in low space. In contrast, for $l \ll \level(U)$ then $S^{(l)}$ is too large to sample or store, so these sets must be compressed or discarded. Naturally, this further complicates the algorithms. We use different techniques in different settings to avoid the aforementioned pitfalls.

\section{Dynamic Algorithm With Insertions and Deletions}
\label{sec:dynamic}

In this section, we prove Theorem~\ref{thm:dynamicalgo} by presenting a fully dynamic (non-streaming) algorithm for union volume estimation in linear space and polylogarithmic update time. Due to Lemma~\ref{lem:reduce-aspect-ratio}, we assume without loss of generality that the volumes of input objects are in the range $[(3n/\eps)^2, (3n/\eps)^4]$.

\subsection{Preparations}

We will use the Karp--Luby--Madras algorithm as a subroutine:

\begin{theorem}[Karp, Luby, Madras~\cite{KarpLM89}]
    \label{thm:KLM}
    Let $n \in \N$ be a parameter. There is an algorithm \KLM{$\class{X}$} that computes a $(1\pm 0.5)$-approximation to $\vol(\bigcup_{X \in \class{X}} X)$ with probability at least $1 - n^{-3}$ and runs in time $O(|\class{X}| \log n)$.
\end{theorem}

\noindent
We need a somewhat technical notion that is ``close'' to an $l$-sample at the right level $l$.
\begin{definition}
    \label{def:digest}
    Let $U$ be a measurable set. A family of pairs $\set{(l_1,R_1), \dots, (l_N, R_N)}$ is called an \emph{ensemble} of $U$ if $l_k \in \level(U) \pm 1$ and $R_k \sim \Pois(U,l_k)$ for all $k \in [N]$.
    A random pair $(L,S)$ is called a \emph{$(\delta,N)$-digest} of $U$ if there exists an ensemble $\class{E}$ of $U$ such that $|\class{E}| \leq N$ and $\Pr[(L,S) \notin \class{E}] \leq \delta$. The ensemble $\class{E}$ is called a \emph{support ensemble} of the digest $(L,S)$.
\end{definition}

\begin{lemma}
    \label{lem:ensemble}
    Fix measurable sets $V \subseteq U$. If $\class{E}$ is an ensemble of $U$, then with probability at least $1 - \frac{2|\class{E}|}{n^4}$, we have $\left| 2^l \card{R \cap V} - \vol(V) \right| \leq \eps \vol(U)$ for all $(l,R) \in \class{E}$.
\end{lemma}

\begin{proof}
    Consider an arbitrary $(l,R) \in \class{E}$. By definition, $l \leq \level(U)+1$ and $R \sim \Pois(U,l)$. In particular, $|R \cap V| \sim \Pois(\lambda)$ where $\lambda := \vol(V)/2^l$. Let us write $d := \eps \vol(U) / 2^l \geq \eps \lambda$. By concentration of Poisson variables~\cite{Canonne19},
    \begin{align*}
        & \Pr\left( \left| 2^l \card{R \cap V} - \vol(V) \right| > \eps \vol(U) \right) \ 
        =\ \Pr( \left| \card{R \cap V} - \lambda \right| > d ) \\
        \leq\ & 2 \exp \left( -\frac{d^2}{2 (\lambda + d)} \right) 
        \leq 2 \exp \left( - \frac{\eps d}{2 (1 + \eps)} \right)
        = 2 \exp \left( - \frac{\eps^2 \vol(U)}{2(1 + \eps)\cdot 2^l} \right)
        \leq 2 \mathrm{e}^{-4 \log n}
        \le \frac{2}{n^4}
    \end{align*}
    where the second last step uses $\frac{\vol(U)}{2^l} \geq \frac{\vol(U)}{2^{\level(U)+1}} \geq \frac{16 \log n}{\eps^2}$. Taking a union bound over all $(l,R) \in \class{E}$ proves the lemma.
\end{proof}

\begin{corollary}
    \label{cor:digest}
    Fix measurable sets $V \subseteq U$. If $S$ is a $(\delta,N)$-digest of $U$, then
    \[ \Pr\left( \left| 2^L \card{S \cap V} - \vol(V) \right| \leq \eps \vol(U) \right)
    \geq 1 - \delta - \frac{2N}{n^4}. \]
\end{corollary}

\begin{proof}
    Let $\class{E}$ be a support ensemble of $(L,S)$, so $|\class{E}| \leq N$ and $\Pr[(L,S) \notin \class{E}] \leq \delta$. The event in the statement is implied by $(L,S) \in \class{E}$ and the event in Lemma~\ref{lem:ensemble}.
\end{proof}

\begin{lemma}
    \label{lem:erase-resample}
    Let $n \in \N$ and $\eps \in (0,\frac{1}{4}]$ be parameters. There is an algorithm \resample{$\class{X}, L$} that, given a multiset $\class{X}$ of $m \leq n$ objects and a random variable $L$, returns a set $S$ of size $|S| = O(\log(n)/\eps^2)$ such that $(L,S)$ is a $(\delta, 3)$-digest of $U := \bigcup_{X \in \class{X}} X$ for $\delta := \frac{2}{n^3} + \Pr[L \not\in \level(U) \pm 1]$. The algorithm runs in time $O(m \log(n) / \eps^2)$.
\end{lemma}

\begin{algorithm}[htb]
    \caption{Digesting a decremental union}
    \label{alg:erase-resample}
    \Function{\resample{$\class{X}, L$}}{
        $S := \emptyset$\;
        \ForEach{$X \in \class{X}$}{
            \ForEach{$x \in S$}{
                \If{$X.\contains(x)$}{
                    remove $x$ from $S$\;
                }
            }
            sample a random number $N \sim \Pois(X.\size/2^L)$\;
            \If{$|S| + N > 160 \log(n) / \eps^2$}{
                \Return $\emptyset$
            }
            \For{$k = 1,\dots,N$}{
                $x := X.\sample()$\;
                add $x$ to $S$\;
            }
        }
        \Return $S$\;
    }
\end{algorithm}

\begin{proof}
    The algorithm is described as Algorithm~\ref{alg:erase-resample}.
    The bound on the size $|S|$ is clear. The running time analysis is also easy: there are $m$ iterations, each running in time $O(\log(n)/\eps^2)$ due to the bound on $|S|$ (and on $|S|+N$).
    
    Next, we show the correctness. 
    Assume first that the algorithm does not return from the if-clause.
    Let $X_1, \dots, X_m$ be the objects in $\class{X}$ in the order enumerated by the foreach-loop, and write $U := \bigcup_{i=1}^m X_i$.
    We define the random set $X_i^{(L)}$ as the set of points sampled from $X_i$ in the for-$k$-loop; observe that $X_i^{(L)} \sim \Pois(X_i, L)$ for every $i \in [n]$.
    Let $S_i$ be the set $S$ at the end of the iteration for $X = X_i$. Note that with $S_0 := \emptyset$ the algorithm computes $S_i = (S_{i-1} \setminus X_i) \cup X_i^{(L)}$ in each iteration.
    It follows that $S_i = \bigcup_{j=1}^i X_j^{(L)} \setminus (X_{j+1} \cup \cdots \cup X_i)$. Since the random sets $X_j^{(L)} \setminus (X_{j+1} \cup \cdots \cup X_i) \sim \Pois(X_j \setminus (X_{j+1} \cup \cdots \cup X_i), L)$ are $L$-samples with disjoint supports and the union of supports is $\bigcup_{j=1}^i X_j$, we have $S_i \sim \Pois(\bigcup_{j=1}^i X_j, L)$.
    
    Now let us analyze the probability of the algorithm aborting. Observe that the algorithm aborts if and only if $|S_i| > 160 \log(n)/\eps^2$ for some $i \in [n]$. By the monotonicity of the size of the support $\bigcup_{j=1}^i X_j$, we have $\Pr[|S_1| > 160 \log(n)/\eps^2] \le \cdots \le \Pr[|S_n| > 160 \log(n)/\eps^2]$.
    Fix $L \in \level(U) \pm 1$. Since $(1 + \eps) \frac{\vol(U)}{2^L} \leq (1 + \eps) \frac{\vol(U)}{2^{\level(U)-1}} \leq (1+\eps) \frac{128 \log n}{\eps^2} \leq \frac{160 \log n}{\eps^2}$, we have
    \[
        \Pr\left( |S_n| > \frac{160 \log n}{\eps^2} \right)
        \leq \Pr\left( |S_n| > (1 + \eps) \frac{\vol(U)}{2^L} \right)
        \leq \frac{2}{n^4},
    \]
    where the second inequality follows from Lemma~\ref{lem:ensemble} by interpreting $(L,S_n)$ as a size-1 ensemble of~$U$ and plugging in $V := U$.
    By a union bound over all events $|S_i| > 160 \log(n)/\eps^2$, it follows that the algorithm aborts with probability at most $\tfrac 2{n^3}$, if $L \in \level(U) \pm 1$.

    If the algorithm does not abort, then we have shown above that it returns a set $S = S_n \sim \Pois(\bigcup_{j=1}^n X_j, L) = \Pois(U, L)$. 
    It follows that we can compare the distribution of $(L,S)$ with the size-3 ensemble $\mathcal{E} := \{(l, \Pois(U,l)) \colon l \in \level(U) \pm 1\}$. We can bound the error probability $\Pr[(L,S) \not\in \mathcal{E}]$ by the probability that $L \not\in \level(U) \pm 1$ plus the probability that the algorithm aborts assuming $L \in \level(U) \pm 1$, and we have bounded the latter probability by $\tfrac 2{n^3}$. This proves the claim.
\end{proof}

\subsection{Digesting decremental union}
Next we study the decremental setting. That is, we receive a set of objects $\class{X}$ and want to build a data structure that allows efficient deletions. At any time, it needs to quickly report a digest of the union of the surviving objects (i.e., objects that have not been deleted).

\begin{algorithm}[htb]
    \caption{Digest data structure for decremental union}
    \label{alg:decremental}
    \Parameter{$n, \eps$}
    \Public{$\class{X}, S, L$}
    \Function{\initialize{$\class{Y}$}}{
        let $\class{X} := \class{Y}$ and $L := \infty$\;
        \refresh{}\;
    }

    \Function{\refresh}{
        $v :=$ \KLM{$\class{X}$}\;
        let $L'$ be the unique integer such that $\frac{32 \log n}{\eps^2} \leq v/2^{L'} < \frac{64 \log n}{\eps^2}$\;
        $L := \min \set{L-1, L'}$\;
        $S := $ \resample{$\class{X}, L$}\;
        \ForEach{$x \in S$}{
            compute and store $m_x := \#\{ X \in \class{X} : X.\contains(x) \}$\;
        }
    }

    \Function{\del{$X$}}{
        $\class{X} := \class{X} \setminus X$\;
        \ForEach{$x \in S$}{
            \If{$X.\contains(x)$}{
                $m_x := m_x - 1$\;
            }
            \If{$m_x = 0$}{
                $S := S \setminus \{x\}$ and forget $m_x$\;
            }
        }
        \If{$|S| < \frac{24 \log n}{\eps^2}$}{
            \refresh{}\;
            \Return true\;
        }\Else{
            \Return false\;
        }
    }
\end{algorithm}

\begin{lemma}
    \label{lem:decremental}
    Let $n \in \N$ and $\eps \in (0,\frac{1}{16}]$ be parameters. 
    Algorithm~\ref{alg:decremental} supports \initialize{$\class{X}$}, which initializes with a multiset $\class{X}$ of $m \le n$ objects whose volumes are in $[(3n/\eps)^2, (3n/\eps)^4]$, and \del{$X$}, which removes object $X$ from $\class{X}$. The algorithm maintains a pair $(L,S)$ which is a $(\frac{18}{n^2}, 6n)$-digest of $\bigcup_{X \in \class{X}} X$ at any point in time, and $|S| = O(\log(n)/\eps^2)$.
    In total, \initialize and any sequence of \del operations run in expected time $O(m \log^2(n/\eps) / \eps^2)$. The space complexity is $O(m + \log(n) / \eps^2)$. 
    
    Moreover, the set $S$ only rarely gains new points: \del{$X$} returns true if the set $S$ after the deletion is not a subset of the set $S$ before the deletion (and false otherwise); and over any sequence of deletions only $O(\log(n/\eps))$ deletions return true, with probability $1 - O(n^{-2})$.
\end{lemma}

\begin{proof}
    Fix the initial objects and an arbitrary sequence of $m$ deletions. After $t = 0, 1, \dots, m$ deletions, let $U_t$ be the union of the surviving objects, and let $L_t,S_t$ be the variables $L,S$ in the algorithm. We claim that $(L_t,S_t)$ is a $(\frac{9(t+1)}{n^3}, 3(t+1))$-digest of $U_t$.
    
    Let us prove the claim by induction. First consider initialization ($t = 0$). By Theorem~\ref{thm:KLM}, with probability at least $1 - n^{-3}$ \KLM computes $v \in (1 \pm 0.5) \vol(U_0)$, and in this case the algorithm computes $L \in \level(U_0) \pm 1$. Thus, \resample returns a $(\frac{3}{n^3},3)$-digest $(L_0,S_0)$ of $U_0$ by Lemma~\ref{lem:erase-resample}.
    
    Now we proceed to $t \geq 1$. By the induction hypothesis, $(L_{t-1},S_{t-1})$ is a $(\frac{9t}{n^3}, 3t)$-digest of $U_{t-1}$. Let $\class{E}$ be a support ensemble of $(L_{t-1}, S_{t-1})$. Define
    \[
        \class{E}' := \set{(l,R \cap U_t) : (l,R) \in \class{E},\, l \leq \level(U_t)+1 }.
    \]
    Observe that $|\class{E}'| \leq |\class{E}| \leq 3t$ and $\class{E}'$ is an ensemble of $U_t$ since $\level(U_t) \leq \level(U_{t+1})$ .
    
    Let $E$ be the event that $(L_{t-1}, S_{t-1}) \in \class{E}$. Let $F$ be the event that $\left| 2^l |R \cap U_t| - \vol(U_t) \right| \leq \eps \vol(U_{t-1})$ for all $(l,R) \in \class{E}$. Note that $\Pr(\overline{E}) \leq \frac{9t}{n^3}$ by the induction hypothesis, and $\Pr(\overline{F}) \leq \frac{2 \cdot 3t}{n^4} < \frac{6}{n^3}$ by Lemma~\ref{lem:ensemble}.
    
    Assuming $E \cap F$, we have $L_{t-1} \in \level(U_{t-1}) \pm 1$ and $2^{L_{t-1}} |S_{t-1} \cap U_t| \in \vol(U_t)\pm\eps \vol(U_{t-1})$. Let us distinguish three cases and derive further consequences:
    \begin{enumerate}[(1)]
        \item If $L_{t-1} \leq \level(U_t)$, then
        \[
            |S_{t-1} \cap U_t|
            \geq \frac{\vol(U_t)}{2^{\level(U_t)}} - \eps \frac{\vol(U_{t-1})}{2^{\level(U_{t-1})-1}}
            \geq \frac{32 \log n}{\eps^2} - \eps \frac{128 \log n}{\eps^2}
            \geq \frac{24 \log n}{\eps^2}.
        \]
        So the algorithm does not refresh. As a result, $L_t = L_{t-1} \leq \level(U_t)$ and $S_t = S_{t-1} \cap U_t$, thus $(L_t, S_t) \in \class{E}'$.

        \item If $L_{t-1} \geq \level(U_t) + 2$, then
        \[
            |S_{t-1} \cap U_t|
            \leq \frac{\vol(U_t)}{2^{\level(U_t)+2}} + \eps \frac{\vol(U_{t-1})}{2^{\level(U_{t-1})-1}}
            < \frac{16 \log n}{\eps^2} + \eps \frac{128 \log n}{\eps^2}
            \leq \frac{24 \log n}{\eps^2}.
        \]
        So the algorithm refreshes. During refresh, the algorithm computes $L' \in \level(U_t) \pm 1$ with probability at least $1 - \tfrac 1{n^3}$ by Theorem~\ref{thm:KLM}, and then we have $L_t = \min \set{L_{t-1}-1, L'} = L' \in \level(U_t) \pm 1$. Hence, it returns a $(\frac{3}{n^3}, 3)$-digest $(L_t, S_t)$ of $U_t$ by Lemma~\ref{lem:erase-resample}.
        
        \item Otherwise $L_{t-1} = \level(U_t) + 1$, and the algorithm may or may not refresh. If it does, then $(L_t,S_t)$ is a $(\frac{3}{n^3}, 3)$-digest of $U_t$ as in case (2). If it does not, then $L_t = L_{t-1} = \level(U_t) + 1$ and $S_t = S_{t-1} \cap U_t$, thus $(L_t, S_t) \in \class{E}'$.
    \end{enumerate}
    Taking all cases into account, recalling $\Pr(\overline{E \cap F}) \leq \Pr(\overline{E}) + \Pr(\overline{F}) \leq \frac{9t+6}{n^3}$, and adding the error probability~$\tfrac{3}{n^3}$ of the new digest after refreshing, we conclude that $(L_t, S_t)$ is a $(\frac{9(t+1)}{n^3}, 3(t+1))$-digest of~$U_t$. (In particular, it is a $(\frac{18}{n^2}, 6n)$-digest of the union of the remaining objects.)

    As a consequence, with probability at least $1 - O(n^{-2})$ we have $L_0 \in \level(U_0) \pm 1$ and $L_{m-1} \in \level(U_{m-1}) \pm 1$. On the other hand, $\level(U_0) \leq O(\log(n/\eps))$ and $\level(U_{m-1}) \geq 0$ (because each object has volume in $[(3n/\eps)^2, (3n/\eps)^4]$), and every refresh causes $L$ to decrease by at least one. So the number of calls to \refresh is bounded by $O(\log(n/\eps))$. Since we have $S_t \subseteq S_{t-1}$ except when we do a refresh, it follows that the set $S$ can only gain points during $O(\log(n/\eps))$ many deletions, with probability $1 - O(n^{-2})$.

    The size bound $|S| = O(\log(n)/\eps^2)$ follows from Lemma~\ref{lem:erase-resample}.

    It remains to bound the running time and space usage.
    For \initialize and \refresh, the subroutine \KLM runs in $O(m \log n)$ time by Theorem \ref{thm:KLM}; the subroutine \resample runs in $O(m \log (n) / \eps^2)$ time by Lemma \ref{lem:erase-resample}; and the loop runs in $O(m \log(n) / \eps^2)$ time since $|S| = O(\log(n)/\eps^2)$ and each point in $S$ is tested against at most $m$ objects.
  
    For \del, we spend $O(\log (n) / \eps^2)$ time in the loop. An extra $O(m \log (n) / \eps^2)$ time may be incurred by a refresh. As we argued before, there are at most $O(\log(n/\eps))$ refreshes throughout the deletion sequence with probability $1 - O(n^{-2})$. There can be at most $m$ refreshes with the remaining probability $O(n^{-2})$. So the expected time for the whole deletion sequence is $O(m \log(n) / \eps^2 \cdot (\log(n/\eps) + m/n^2)) \leq O(m \log^2(n/\eps)/\eps^2)$.

    Finally, let us analyze the space complexity. Storing $\class{X}$ in a binary search tree needs $O(m)$ space. Storing $S$ and the counters $m_x$ in a table needs $O(\log(n)/\eps^2)$ space. Storing $L$ needs constant space. So the algorithm uses $O(m + \log(n)/\eps^2)$ space in total.
\end{proof}

\subsection{Handling insertions}
In this section, we apply the decremental data structure from the previous section to derive a dynamic algorithm for union volume estimation, using a novel variant of the logarithmic method~\cite{BS80}. Throughout the algorithm, the remaining objects are partitioned into bins $\class{X}_0, \dots, \class{X}_{\log n}$. Each bin $\class{X}_j$ holds at most $2^j$ objects and is implemented by the decremental data structure. Every now and then, we merge some (small) bins into a larger bin by re-initializing the target bin with the objects in the source bins. Between any two merges that involve a bin $\class{X}_j$, it undergoes deletions only. At any time, we can access a digest $(\class{X}_j.L, \class{X}_j.S)$ of the union $\bigcup_{X \in \class{X}_j} X$. Additionally, for every point $x \in \class{X}_j.S$, we keep track of the number $n_x := \#\{ X \in \class{X}_0 \cup \cdots \cup \class{X}_{j-1} : x \in X \}$.

\begin{theorem}\label{thm:dynamicintermediate}
    Let $n \in \N$ and $\hat \eps \in (0,\frac{1}{16}]$ be parameters. 
    Algorithm~\ref{alg:dynamic} maintains an initially empty multiset $\class{X}$ of objects, where each object has volume in $[(3n/\hat \eps)^2, (3n/\hat \eps)^4]$, under updates \ins{$X$} and \del{$X$} and query \est{}, which $(1\pm\hat \eps)$-approximates the volume of the union of $\class{X}$ and is correct with high probability. The amortized expected update time is $O(\log^5(n/\hat \eps)/\hat \eps^2)$, the query time is $O(\log^4(n)/\hat \eps^2)$, and the space complexity is $O(n + \log^4(n)/\hat \eps^2)$.
\end{theorem}

\begin{algorithm}[htb]
\caption{Dynamic union volume estimation with arbitrary insertions and deletions}
\label{alg:dynamic}
    \Parameter{$n, \hat \eps$}
    
    \Function{\initialize}{
        \For{$j = 0, \dots, \log n$}{
            $\class{X}_j$.\initialize{$\emptyset$} with parameters $n$ and $\eps := \hat \eps/(1+\log n)$\;
        }
    }

    \Function{\ins{$X$}}{
        let $h$ be the minimal integer such that $1 + \sum_{j = 0}^h | \class{X}_j | \leq 2^h$\;
        $\class{X}_h$.\initialize{$\{X\} \cup \class{X}_0 \cup \cdots \cup \class{X}_h$} \tcc*{merge}
        $\class{X}_j$.\initialize{$\emptyset$} for $j =0, \dots h-1$\;
        \ForEach{$x \in \class{X}_h.S$}{
            $n_x := 0$
        }
        \ForEach{$x \in \bigcup_{j=h+1}^{\log n} \class{X}_j.S$}{
            \If{$X.\contains(x)$}{
                $n_x := n_x + 1$
            }
        }
    }

    \Function{\del{$X$}}{
        find the bin $\class{X}_j$ containing $X$\;
        $\mathrm{refreshed} := \class{X}_j$.\del{$X$}\;
        \If{refreshed}{
            \ForEach{$x \in \class{X}_j.S$}{
                $n_x := \#\{ X \in \class{X}_0 \cup \cdots \cup \class{X}_{j-1} : X.\contains(x) \}$\;
            }
        }
        \ForEach{$x \in \bigcup_{k=j}^{\log n} \class{X}_k.S$}{
            \If{$X.\contains(x)$}{
                $n_x := n_x - 1$
            }
        }
    }

    \Function{\est{}}{
        denote $S_j := \class{X}_j.S$ and $L_j := \class{X}_j.L$ for $j = 0, 1, \dots, \log n$\;
        compute $R_j := \{ x \in S_j : n_x = 0 \}$ for $j = 0, 1, \dots, \log n$\;
        \Return $\sum_{j=0}^{\log n} 2^{L_j} |R_j|$\;
    }
\end{algorithm}

\begin{proof}
    Fix a sequence of at most $n$ insertions and deletions. Then the partition of surviving objects into bins $\class{X}_1, \dots, \class{X}_{\log n}$ is determined; in particular $U_j := \bigcup_{X \in \class{X}_j} X$ is determined, and we write $V_j := U_j \setminus (U_0 \cup \cdots \cup U_{j-1})$. Note that the union of the surviving objects $U := \bigcup_{j=0}^{\log n} U_j = \bigcup_{j=0}^{\log n} V_j$ is a disjoint union of $V_j$'s, thus $\vol(U) = \sum_{j=0}^{\log n} \vol(V_j)$.
    
    Now let us consider a call to \est. By Lemma~\ref{lem:decremental}, each $S_j$ is a $(\frac{18}{n^2}, 6n)$-digest of $U_j$. Note that the algorithm computes $R_j := V_j \cap S_j$, so by Corollary~\ref{cor:digest}, we have $2^{L_j} |R_j| \in \vol(V_j) \pm \eps \vol(U_j)$ with probability at least $1 - \frac{18}{n^2} - \frac{12n}{n^4} = 1 - O(n^{-2})$. By a union bound, this holds for all $j$ simultaneously with probability $1 - O(n^{-1})$. Under this event,
    \[ \sum_{j=0}^{\log n} 2^{L_j}|R_j| \in \vol(U) \pm \eps \sum_{j=0}^{\log n} \vol(U_j). \]
    We can crudely bound the additive error by $\eps (1+\log n)\vol(U) = \hat \eps \vol(U)$, just as desired.

    Next we analyze the running time. We will charge the cost of \del to insertions. Note that in \del, it takes time $O(\log^2(n))$ to compute the bin $j$ containing $X$, if for each $j$ we store (pointers to) all objects in $\class{X}_j$ in a binary search tree. The last for-loop takes time $O(\log^2(n) / \eps^2)$. To bound the running time of the ``if refreshed'' block, we change our viewpoint and focus on the sequence of deletions occurring between two merges of a particular bin $\class{X}_j$.
    \begin{itemize}
        \item The total cost of $\class{X}_j.\del$ is $O(2^j \log^2(n/\eps) / \eps^2)$ in expectation by Lemma~\ref{lem:decremental}.  
        \item With probability $1 - O(n^{-2})$, there are $O(\log(n/\eps))$ refreshes between two merges of bin $\class{X}_j$ by Lemma~\ref{lem:decremental}, so the if-block is executed $O(\log(n/\eps))$ times. Each time, the algorithm recomputes $n_x$ for every $x \in S_j$ by testing against at most $\sum_{k = 0}^{j-1} 2^k \leq 2^j$ objects, which takes time $O(2^j \log(n) / \eps^2)$. Therefore, the total contribution of the if-block is $O(2^j \log^2(n/\eps) / \eps^2)$ in expectation.
    \end{itemize}
    Therefore, over the sequence of deletions between two merges of $\class{X}_j$, deletions from $\class{X}_j$ take expected total time $O(2^j \log^2(n/\eps) / \eps^2)$. Note that at least $2^j$ insertions occur between two merges of bin $\class{X}_j$. We can thus charge the cost of deletions to insertions. Each insertion is charged $O(\log n)$ times (once for every bin), so we charge each insertion with a cost of $O(\log^3(n/\eps) / \eps^2)$ in expectation.

    Now we analyze the running time of \ins. It depends on the value $h$. By Lemma \ref{lem:decremental}, the merge step costs $O \left( \sum_{j = 0}^h 2^j \log^2(n/\eps) / \eps^2 \right) = O \left( 2^h \log^2(n/\eps) / \eps^2 \right)$, which can be expensive in the worst case. However, bin $\class{X}_h$ is merged at most once every $2^h$ updates, so it contributes amortized cost $O(\log^2(n/\eps) / \eps^2)$ to each insertion. Summing over all $h$, the amortized cost of merges is $O(\log^3(n/\eps) / \eps^2)$ per insertion.
    Next we consider the foreach-loops. Since each digest contains $O(\log(n) / \eps^2)$ points by Lemma~\ref{lem:decremental}, and the algorithm spends constant time on each point, these loops take time $O(\log^2(n) / \eps^2)$.
    To summarize, \ins runs in amortized expected time $O(\log^3(n/\eps) / \eps^2) = O(\log^5(n/\eps) / \hat \eps^2)$ (and \del runs in amortized time 0).
    
    The running time of \est is clearly $O(\log^2(n) / \eps^2) = O(\log^4(n)/\hat \eps^2)$.
    
    The space of bin $\class{X}_j$ is $O(2^j + \log(n) / \eps^2)$ by Lemma~\ref{lem:decremental}, and we can charge the space of counters $n_x$ to it. In total the space complexity is $O\left( \sum_{j=0}^{\log n} (2^j + \log(n) / \eps^2) \right) = O(n + \log^4(n) / \hat \eps^2)$.
\end{proof}

Theorem~\ref{thm:dynamicalgo} follows by combining Theorem~\ref{thm:dynamicintermediate} with the ``moreover'' part in Lemma~\ref{lem:reduce-aspect-ratio}. Neither the space bound nor the time bound depends on the aspect ratio of the input.

\section{Streaming Algorithm With Insertions and Suffix Queries}
\label{sec:window}

This section presents our streaming algorithm with insertions and suffix queries (Algorithm~\ref{alg:window}).
In principle, we could use the language of digests (Definition~\ref{def:digest}) to analyze it. But here we choose to apply a simpler lemma that avoids the heavy machinery:
\begin{lemma}
    \label{lem:select-sample}
    Let $U$ be a measurable set with $\level(U) \in \{a,a+1,\ldots,b\}$ and consider (possibly dependent) random sets $S^{(a)}, S^{(a+1)}, \dots, S^{(b)}$ where each $S^{(l)}$ is an $l$-sample of $U$. Define a random variable $L := \min \set{l \in \{a,\ldots,b\} : |S^{(l)}| \leq \frac{100 \log n}{\eps^2}}$. Then with probability $1 - O(n^{-1})$ we have $\level(U)-2 \leq L \leq \level(U)$ and $2^L |S^{(L)}| \in (1 \pm \eps) \vol(U)$.
\end{lemma}

\begin{proof}
    For each $a \leq l \leq \level(U)$, define two events
    \[
        E_l := \set{2^l |S^{(l)}| \in (1 \pm \eps) \vol(U)} \quad\text{and}\quad
        F_l := \set{2^l |S^{(l)}| \in (1 \pm 0.5) \vol(U)}.
    \]
    Then $\Pr(\overline{E_l}) \leq O(n^{-8})$ by Lemma~\ref{lem:l-sample}. Similarly as in Lemma~\ref{lem:l-sample} we can bound 
    \begin{align*}
        \Pr(\overline{F_l})
        &= \Pr \left( \left| 2^l |S^{(l)}| - \vol(U) \right| > \tfrac 12 \vol(U) \right) \\
        &\leq 2\exp\left( \frac{-(1/2)^2 \vol(U)}{4 \cdot 2^l} \right)
        \leq 2\exp\left( \frac{-2 \log(n) \cdot 2^{\level(U)}}{2^l \cdot \eps^2} \right)
        \leq O(n^{-2^{\level(U)-l}}),
    \end{align*}
    where we used the inequality $\frac{32 \log n}{\eps^2} \leq \frac{\vol(U)}{2^{\level(U)}}$ from the definition of $\level(U)$, and $\eps \le 1$. A union bound now shows that all events $F_l$ for $l = a,\ldots,\level(U)$ simultaneously hold with probability at least $1 - O(n^{-1})$. Hence, we have
    \[
        \Pr\left( \bigcap_{l=\level(U)-2}^{\level(U)} E_l \cap \bigcap_{l=a}^{\level(U)} F_l \right)
        \geq 1 - O(n^{-1}).
    \]
    Conditioned on this event,
    \begin{itemize}
        \item For all $a \leq l \leq \level(U) - 3$, $\card{S^{(l)}} \geq 0.5 \cdot \frac{\vol(U)}{2^l} \geq 0.5 \cdot 8 \cdot \frac{32 \log n}{\eps^2} = \frac{128 \log n}{\eps^2}$. Hence $L \geq \level(U) - 2$.
        \item $\card{S^{\level(U)}} \leq 1.5 \cdot \frac{\vol(U)}{2^{\level(U)}} \leq 1.5 \cdot \frac{64 \log n}{\eps^2} \leq \frac{96 \log n}{\eps^2}$. Hence $L \leq \level(U)$.
        \item $2^l |S^{(l)}|$ yields the desired approximation for each $\level(U) - 2 \le l \le \level(U)$.
    \end{itemize}
    These together prove the claim.
\end{proof}

\begin{algorithm}[htb]
    \caption{Streaming algorithm for union volume estimation}
    \label{alg:window}
    \Parameter{$n, \eps$}
    \Function{\initialize{}}{
        $S^{(l)} := \emptyset$ and $s^{(l)} := 0$ for all $0 \leq l \leq 5 \log(3n/\eps)$\;
        $t := 0$\;
        $C := 100 \log(n)/\eps^2$\;
    }
    \Function{\ins{$X$}}{
        $t := t + 1$\;
        \For{$l = 0, \dots, 5 \log(3n/\eps)$}{
            \ForEach{$x \in S^{(l)}$}{
                \If{$X.\contains(x)$}{
                    remove $x$ from $S^{(l)}$
                }
            }
            sample $N \sim \Pois(X.\size / 2^l)$\;
            \If{$N \le C$}{
                binary search for the smallest $a$ with $N + |\{ x \in S^{(l)} \colon \mathrm{time}(x) \ge a \}| \le C$\;
                remove all $x \in S^{(l)}$ with $\mathrm{time}(x) < a$\;
                $s^{(l)} := a$\;
                \For{$i = 1, \dots N$}{
                    $x := X.\sample()$\;
                    set $\mathrm{time}(x) := t$ and add $x$ to $S^{(l)}$\;
                }
            }\Else{
                $s^{(l)} := t+1$, $S^{(l)} := \emptyset$\;
            }
        }
    }
    \Function{\est{$s$}}{
        $L := \min \set{l : s^{(l)} \leq s}$\;
        $R := \set{x \in S^{(L)} : \mathrm{time}(x) \geq s }$\;
        \Return $2^L |R|$\;
    }
\end{algorithm}

\begin{theorem}
\label{thm:windowalgo2}
    Algorithm~\ref{alg:window} implicitly maintains a sequence of objects $X_1, \dots, X_t$ where each object has volume in $[(3n/\eps)^2, (3n/\eps)^4]$. The algorithm supports update \ins{$X$}, which is given oracle access to the object $X$ and appends $X$ (as $X_{t+1}$) to the sequence, and query \est{$s$}, which outputs a $(1\pm\eps)$-approximation to $\vol(X_s \cup \cdots \cup X_t)$ with high probability. The update time and space complexity are $O(\log(n/\eps) \log(n)/\eps^2)$, and the query time is $O(\log(n)/\eps^2)$.
\end{theorem}

\begin{proof}
    Since each object has volume in $[(3n/\eps)^2, (3n/\eps)^4]$, their union has volume in $[(3n/\eps)^2, (3n/\eps)^5]$, so the level of the union (as in Definition~\ref{def:level}) is in $[0, 5\log(3n/\eps)]$. In particular, the correct level is considered by some choice of $l$ in the algorithm.
    
    Let us denote $U[a,b] := \bigcup_{i=a}^b X_i$, with the understanding that $U[a,b] = \emptyset$ for $a > b$. For the sake of analysis, we define the following random sets for $l = 0, \dots, 5 \log(3n/\eps)$:
    \begin{itemize}
        \item $X_i^{(l)} \sim \Pois(X_i, l)$ for every $i \in [n]$.
        \item $U^{(l)}[a,b] := \bigcup_{i=a}^{b} ( X_i^{(l)} \setminus U[i+1,b] )$ for every $a \leq b$.
    \end{itemize}
    Observe that $U^{(l)}[a,b]$ is a union of independent $l$-samples $X_i^{(l)} \setminus U[i+1,b] \sim \Pois(X_i \setminus U[i+1,b], l)$, whose supports are mutually disjoint, and the union of supports is $U[a,b]$. Therefore $U^{(l)}[a,b] \sim \Pois(U[a,b], l)$. Also note that $U^{(l)}[1,t] \supseteq \cdots \supseteq U^{(l)}[t,t]$.

    Pretending that the for-$i$-loop of \ins is always executed, we may interpret $X_i^{(l)}$ as the set of points sampled at iteration $l$ of \ins{$X_i$}. We claim two invariants over time:
    \begin{enumerate}[(i)]
        \item $s^{(l)} = \min \set{a : \card{U^{(l)}[a,t]} \leq C }$;
        \item $S^{(l)} = U^{(l)}[s^{(l)}, t]$.
    \end{enumerate}
    
    The claims apparently hold after initialization ($t = 0$), so we focus on the call \ins{$X_i$} at time $t \geq 1$. Consider an arbitrary iteration $l$ of the outer loop. By the invariant, at the beginning of iteration $l$ we have $S^{(l)} = U^{(l)}[s^{(l)}, t-1]$. The foreach-loop removes all points in $X_t$, so at the end of the foreach-loop we have $S^{(l)} = U^{(l)}[s^{(l)}, t-1] \setminus X_t$. Next $N \sim \Pois(X_t.\size/2^l)$ is sampled, and we interpret it as determining $|X_t^{(l)}|$ beforehand---the actual $X_t^{(l)}$ is realized later in the for-$i$-loop. 
    
    Observe that for any $s^{(l)} \le a \le t$ we have
    \begin{align*}
        U^{(l)}[a,t]
        &= X_t^{(l)} \cup (U^{(l)}[a, t-1] \setminus X_t) \\
        &= X_t^{(l)} \cup ((U^{(l)}[s^{(l)}, t-1] \cap U[a, t-1]) \setminus X_t) \\
        &= X_t^{(l)} \cup ((U^{(l)}[s^{(l)}, t-1] \setminus X_t) \cap U[a, t-1]) \\
        &= X_t^{(l)} \cup (S^{(l)} \cap U[a,t-1])  \\
        &= X_t^{(l)} \cup \{x \in S^{(l)} \colon \mathrm{time}(X) \ge a \},
    \end{align*}
    and thus $|U^{(l)}[a,t]| = N + |\{x \in S^{(l)} \colon \mathrm{time}(X) \ge a \}|$. It follows that the binary search in the if-block correctly updates $s^{(l)}$ according to invariant (i). It also updates $S^{(l)}$ to still satisfy $S^{(l)} = U^{(l)}[s^{(l)}, t-1] \setminus X_t$ (for the new value of $s^{(l)}$). The for-$i$-loop now samples $X_t^{(l)}$ of size $N = |X_t^{(l)}|$ and adds it to $S^{(l)}$, so we obtain $S^{(l)} = U^{(l)}[s^{(l)}, t]$, so invariant (ii) holds. Correctness in the else-case also clearly holds.

    The running time of \ins is easy to bound.
    Consider a fixed iteration $l$ of some insertion. The foreach-loop takes time $O(|S^{(l)}|) = O(C)$, due to invariants (i) and (ii). 
    If we store the set $S^{(l)}$ in a binary search tree ordered by $\mathrm{time}(x)$, then the binary search takes time $O(\log n)$, and in the same time bound we can remove all items with $\mathrm{time}(x) < a$.
    The for-$i$-loop takes time $O(N)$ and is only executed if $N \le C$, so it runs in time $O(C)$.
    As there are $O(\log(n/\eps))$ iterations over~$l$, \ins runs in time $O(\log(n/\eps) \cdot (C + \log n)) = O(\log(n/\eps) \log(n)/\eps^2)$.

    \smallskip
    To show the correctness of \est{$s$} at time $t$ we make two more claims:
    \begin{enumerate}[(i)]
        \setcounter{enumi}{2}
        \item $L = \min \set{ l : \card{U^{(l)}[s,t]} \leq C }$;
        \item $R = U^{(L)}[s,t]$.
    \end{enumerate}
    Inspecting \est{$s$} we have $L = \min \set{l : s^{(l)} \leq s}$. We have $s^{(l)} \leq s$ if and only if $\card{U^{(l)}[s,t]} \leq C$, due to (i) and the monotonicity $\card{U^{(l)}[1,t]} \geq \cdots \geq \card{U^{(l)}[s,t]}$; this proves (iii). Inspecting \est{$s$} we have $R = \set{x \in S^{(L)} : \mathrm{time}(x) \geq s } = S^{(L)} \cap U[s,t]$. Since $S^{(L)} = U^{(L)}[s^{(L)},t]$ by (ii), and $s^{(L)} \leq s$, we obtain (iv).
    
    With the claims established, we apply Lemma \ref{lem:select-sample} on the (dependent) random sets $U^{(l)}[s,t]$ for $0 \le l \le 5 \log(3n/\eps)$. Due to (iii), our $L$ is defined exactly as in the lemma. The lemma implies that $2^L |U^{(L)}[s,t]|$ is a $(1 \pm \eps)$-approximation to $\card{U[s,t]}$ with high probability. Since $|U^{(L)}[s,t]| = |R|$ by (iv), we return a correct estimate with high probability.

    Clearly, the running time of \est is proportional to $O(\log(n/\eps) + C) \leq O(\log(n)/\eps^2)$.

    Finally, the space complexity is $O(\log(n/\eps) \cdot C) = O(\log(n/\eps) \log(n)/\eps^2 )$.
\end{proof}

Theorem~\ref{thm:windowalgo} now follows by combining Theorem~\ref{thm:windowalgo2} with Lemma~\ref{lem:reduce-aspect-ratio}.

We point out that Algorithm~\ref{alg:window} specializes to an insertion-only streaming algorithm if we hardwire $s = 1$ in the queries. In fact, the streaming algorithm by Meel, Vinodchandran, and Chakraborty \cite{MVC21} can be viewed as a space-optimized version of Algorithm~\ref{alg:window}. Some subtle inaccuracies of their proof were addressed in our proof above.

\section{Streaming Algorithm for Convex Bodies With Insertions and Deletions}
\label{sec:convex}
In this section we present a streaming algorithm for convex bodies in $\R^d$ under insertions and deletions. Throughout this section we assume that $d$ is constant. We also assume the following input condition in line with the literature on volume estimation. Let $R \geq r > 0$ be given. We assume that every object is convex, is contained in $[0,R]^d$, and contains a ball of radius $r$. The balls for different objects may differ, and we do not know them a priori. The objects are accessed via oracles as before.

For $\alpha > 0$ and a convex body $X$ with center $o$, we define the \emph{$\alpha$-dilation} of $X$ by $\alpha X := \{ o + \alpha (x-o) : x \in X \}$. Unless otherwise stated, we take $o$ as the center of gravity.

\subsection{Reduction to point counting}
Let us introduce a related problem called \emph{union point counting}: Given a parameter $\Delta \in \N$ and a multiset of objects $\class{X}$ in $\R^d$ in the oracle model, compute a $(1 \pm \eps)$-approximation to the cardinality $\card{[\Delta]^d \cap \bigcup_{X \in \class{X}} X}$. Just like union volume estimation, we study this problem in the dynamic setting where $\class{X}$ undergoes insertions and deletions.

As the parameter $\Delta$ is usually implied from the context, we denote $\discrete{U} := [\Delta]^d \cap U$. To relate union point counting and union volume estimation, we need a lemma about convexity.

\begin{lemma}
    \label{lem:grid}
    Let $X \subset \R^d$ be a convex body that contains a ball of radius $r > 0$. Assume $0 < \delta \leq r/d$ and let $\partial_\delta X := \set{x \in \R^d : \mathrm{dist}(x,\partial X) \leq \delta}$ be the $\delta$-neighborhood of the boundary $\partial X$. Then we have $\vol(\partial_\delta X) \leq (2 \mathrm{e} d \delta / r) \cdot \vol(X)$.
\end{lemma}

\begin{proof}
    We take the center of the $r$-ball in $X$ as the center of $X$. Consider the sets $S_\text{in} := \partial_\delta X \cap X$ and $S_\text{out} := \partial_\delta X \setminus X$. It is obvious that $\vol(S_\text{in}) \leq \vol(S_\text{out})$ because $X$ is convex. (For a proof, it suffices to consider convex polytopes as they can approximate convex bodies within arbitrary precision. For each cone formed by the center and a facet of the polytope, the outer volume is larger than the inner volume.) So it remains to bound $\vol(S_\text{out})$.

    To this end, observe that $S_\text{out} \subseteq (X + (\delta/r) X) \setminus X$ where the addition is the Minkowski sum. Indeed, $(\delta/r) X$ contains a $\delta$-ball by definition, so the Minkowski sum covers the neighborhood $\partial_\delta X$.

    Since $X$ is convex, we have $X + (\delta/r) X = (1+\delta/r) X$. Hence, we can bound
    \begin{align*}
        \vol(S_\text{out})
        &\leq \vol((1+\delta/r)X) - \vol(X) \\
        &= ((1+\delta/r)^d - 1) \vol(X) \\
        &\leq (1+\delta/r)^d \left( 1 - (1 - \delta/r)^d \right) \vol(X) \\
        &\leq (1+\delta/r)^d \cdot (d \delta / r) \cdot \vol(X) \\
        &\leq (\mathrm{e} d \delta/r) \cdot \vol(X),
    \end{align*}
    where the second last step applied a union bound, and the last step used $\delta \leq r/d$ and the inequality $1+x \leq e^x$. The lemma follows after taking $S_\text{in}$ into account.
\end{proof}

\begin{lemma}
    \label{lem:discretization}
    Assume that each input object is a convex body in $[0,R]^d$ and contains a ball of radius $r>0$. If there is a dynamic algorithm for union point counting with update and query time $T(n,\Delta,\eps)$ and space $S(n,\Delta,\eps)$, then there is a dynamic algorithm for union volume estimation with update and query time $O(T(n, \Delta', \frac{\eps}{3}))$ and space $O(S(n, \Delta', \frac{\eps}{3}))$ where $\Delta' = O(\frac{n R}{\eps r})$.
\end{lemma}

\begin{proof}
    We scale up the space (and implicitly the objects) by a factor of $\lambda := \frac{18 d^{3/2} n}{\eps r}$ in each dimension. So now each object is a convex body in $[0, \lambda R]^d$ and contains a ball of radius $\lambda r = \frac{18 d^{3/2} n}{\eps}$. We subdivide the space into unit cubes such that the centers of the cubes lie on the integer grid $[\Delta]^d$, where $\Delta := \lambda R + 1$. We build a dynamic data structure for union point counting on $[\Delta]^d$ with precision parameter $\eps/3$. The insertions and deletions of objects are directly passed to the data structure. To estimate the union volume of the surviving objects, we query the data structure and obtain a count $Z$. We output $Z/\lambda^d$.

    The time and space complexity are clearly as stated. To show correctness, let $X_1, \dots, X_n$ be the scaled objects currently in $\class{X}$. For each $i \in [n]$, let $X_i'$ be the union of cubes whose centers are in $X_i$. Denote $U := \bigcup_{i=1}^n X_i$ and $U' := \bigcup_{i=1}^n X_i'$. Since the counting data structure is correct, $Z \in (1\pm\frac{\eps}{3}) |\discrete{U}|$. On the other hand, $|\discrete{U}| = \vol(U')$. It remains to relate $\vol(U)$ and $\vol(U')$. To this end, we write $\oplus$ for the symmetric difference and crudely bound
    \[
        \vol(U \oplus U')
        = \vol \left( \bigg( \bigcup_{i=1}^n X_i \bigg) \oplus \bigg( \bigcup_{i=1}^n X'_i \bigg) \right)
        \leq \sum_{i=1}^n \vol(X_i \oplus X'_i). \tag{$\star$}
    \]

    We claim that $\vol(X_i \oplus X_i') \leq \frac{\eps}{3n} \vol(X_i)$ for all $i \in [n]$. Indeed, the symmetric difference $X_i \oplus X_i'$ is contained in the neighborhood $\partial_\delta X_i$ for $\delta := \sqrt{d}$. Since $\delta \leq (\lambda r)/d$, we can apply Lemma~\ref{lem:grid} to bound $\vol(\partial_\delta X_i) \leq \frac{2 \mathrm{e} d \delta}{\lambda r} \vol(X_i) < \frac{\eps}{3n} \vol(X_i)$. Now we can continue to bound
    \[
        (\star)
        \leq \frac{\eps}{3n} \sum_{i=1}^n \vol(X_i)
        \leq \frac{\eps}{3n} \sum_{i=1}^n \vol(U)
        = \frac{\eps}{3} \vol(U).
    \]
    Therefore, $\vol(U') \in (1 \pm \frac{\eps}{3}) \vol(U)$, and consequently $Z \in (1\pm \frac{\eps}{3})^2 \vol(U) \subset (1\pm\eps) \vol(U)$.
\end{proof}

Due to Lemma~\ref{lem:discretization}, we focus on solving union point counting in what follows.

\subsection{Preparations}

\begin{theorem}[John ellipsoid \cite{John48}]
    \label{thm:john-ellipsoid}
    Let $X \subset \R^d$ be a convex body. There exists a unique ellipsoid $E \subseteq X$ of maximum volume, and it satisfies $X \subseteq 2d^{3/2} E$. Similarly, there exists a unique ellipsoid $E' \supseteq X$ of minimum volume, and it satisfies $X \supseteq \frac{1}{2d^{3/2}} E'$. The sets $E$ and $E'$ are called the inner and outer \emph{John ellipsoids} of $X$, respectively.
\end{theorem}

\begin{theorem}[Welzl \cite{Welzl05}]
    \label{thm:min-ellipsoid}
    There is an algorithm that computes the minimum enclosing ellipsoid of $k$ points in $\R^d$ in $O(k)$ expected time for constant $d$.
\end{theorem}

\begin{corollary}
    \label{cor:bounding-box}
    Given a convex body $X \subset \R^d$, we can compute in expected time $O(\log n)$ a rotated box $B \supseteq X$ such that $\vol(B) \leq (12d^4)^d \vol(X)$.
\end{corollary}

\begin{proof}
    The algorithm is simple: Sample $k := (6d^2)^d (d + 2\log n)$ random points $x_1, \dots, x_k \in X$, compute the minimum enclosing ellipsoid $F$ using Theorem~\ref{thm:min-ellipsoid},\footnote{
    As an alternative to Welzl's algorithm, we can apply Barequet and Har-Peled's approximate bounding box algorithm~\cite{BarequetH01} to these $k$ random points.}
    then output the minimum bounding box $B$ of the dilated ellipsoid $6d^2 F$.

    The running time is clearly $O(k) = O(\log n)$ since $d$ is a constant. To show correctness, let $E$ be the inner John ellipsoid of $X$, and let $G$ be a box of maximum volume contained in $E$. (If $E$ is a ball, then $G$ is a cube. In general $G$ is a cube stretched along the axes of $E$.) We divide each side of $G$ evenly into three segments. This induces $3^d$ grid cells in $G$, of which $2^d$ cells are at the corners. We denote the cell at the center by $G_0$. See Figure~\ref{fig:convex-approx}.
    
    \begin{figure}[htb]
        \centering
        \includegraphics{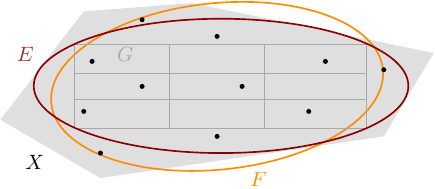}
        \caption{The convex body $X$, the inner John ellipsoid $E$, its inscribing box $G$, the subdivision into cells, and the minimum ellipsoid $F$ of the sample points.}
        \label{fig:convex-approx}
    \end{figure}

    Note that $G = 3 G_0$ and $E \subseteq \sqrt{d} G$ by definition. Additionally, $X \subseteq 2d^{3/2} E$ by Theorem~\ref{thm:john-ellipsoid}. Chaining these together, we obtain $X \subseteq 6 d^2 G_0$. In particular, the volume of each cell is at least $\rho := \frac{1}{(6d^2)^d}$ times the volume of $X$.

    Now consider a specific corner cell. For each $i \in [k]$, the sample point $x_i$ hits the cell with probability at least $\rho$. So the probability that no sample point hits the cell is at most
    \[ (1 - \rho)^k \leq \exp(-\rho k) = \mathrm{e}^{-(d + 2 \log n)} = \mathrm{e}^{-d} \cdot n^{-2}. \]
    Applying a union bound over all corner cells, we have with probability at least $1-n^{-2}$ that all corner cells are hit. Conditioned on this event, $G_0 \subseteq F$ and thus $X \subseteq 6d^2 F \subseteq B$.

    It remains to relate $\vol(B)$ to $\vol(X)$. To this end, we consider the outer John ellipsoid $E'$ of $X$. We have $\vol(E') \leq (2d^{3/2})^d \vol(X)$ by Theorem~\ref{thm:john-ellipsoid}. Next note that $\vol(F) \leq \vol(E')$, for otherwise $E'$ would be a smaller ellipsoid containing all sample points, contradicting the definition of $F$. Finally, $B \subseteq \sqrt{d} \cdot (6d^2 F)$ by definition, thus $\vol(B) \leq (6d^{5/2})^d \vol(F)$. Chaining these together, we conclude that $\vol(B) \leq (12d^4)^d \vol(X)$.
\end{proof}

\begin{theorem}[Lenstra~\cite{Lenstra83}]\!\!\footnote{There were subsequent improved algorithms, see e.g., \cite{ReisR23}, but we don't need them here.}
    \label{thm:ILP}
    Let $k \in \N$ be a constant. There is an algorithm that solves any given integer linear program with $k$ variables and description size $m$ in time $O(m^{c_k})$, where $c_k$ is a constant depending only on $k$.
\end{theorem}

\begin{corollary}
    \label{cor:ILP}
    Let $k \in \N$ be a constant. There is an algorithm that lists all feasible solutions of any given linear program with $k$ variables and description size $m$ in time $O(z \cdot m^{c_k'})$, where $z$ is the output size and $c_k'$ is a constant depending only on $k$.
\end{corollary}

\begin{proof}
	We describe a recursive algorithm \ILP{$P,k$} that lists all integral points inside a $k$-dimensional polytope $P$. Denoting by $P[x_i \mapsto a]$ the polytope $P$ after fixing the $i$-th coordinate to a constant $a$, the algorithm is given as Algorithm~\ref{alg:ILP}.
    
	\begin{algorithm}[htbp]
        \caption{Listing all solutions in an ILP}
        \label{alg:ILP}
        \Function{\ILP{$P, k$}}{
            \If{$k = 0$}{
                output the unique point in $P$\;
            }\Else{
                $a := \min \set{x_k : x \in P \cap \Z^k}$\;
                \While{$a < \infty$}{
                    \ILP{$P[x_k \mapsto a], k-1$}\;
                    $a := \min \set{x_k : x \in P \cap \Z^k, ~x_k \geq a+1}$\;
                }
            }
        }
    \end{algorithm}

    We show by induction on $k$ that \ILP{$P,k$} is correct and runs in time $O_k(z \cdot m^{c_k'})$, where $O_k(\cdot)$ hides constant factors depending on $k$. In the base case $k = 0$, the polytope $P$ is a singleton and thus $z = 1$. The algorithm correctly outputs the unique point in $P$ in time $O(m) = O(z \cdot m)$.

    Now consider $k \geq 1$. Denote by $a_1 < \dots < a_r < \infty$ the evolution of $a$ during the while-loop. Observe that $a_1, \dots, a_r$ enumerates all feasible integer assignments of the last coordinate $x_k$. For each $t \in [r]$, the algorithm recurses on the non-empty subpolytope $P_t := P[x_k \mapsto a_t]$. By induction, this lists all (say $z_t \geq 1$ many) integral points in $P_t$. Correctness follows since the polytopes $P_1, \dots, P_r$ are disjoint and cover all integral points of $P$. Moreover, we have $z = \sum_{t=1}^r z_t$. To bound the running time, for each $t \in [r]$ the recursive call costs time $O_k(z_t \cdot (m+1)^{c_{k-1}'}) = O_k(z_t \cdot m^{c_{k-1}'})$ by induction, and updating $a$ costs time $O(m^{c_k})$ by Theorem~\ref{thm:ILP}. So the total time is
	\[ \sum_{t=1}^r \left( O(m^{c_k}) + O_k(z_t \cdot m^{c_{k-1}'}) \right) = O_k(r \cdot m^{c_k} + z \cdot m^{c_{k-1}'}). \]
	Note that $r \leq z$ since $z_1, \dots, z_r \geq 1$. So the time is $O_k(z \cdot m^{c_k'})$ where $c_k' := \max(c_k, c_{k-1}')$.
\end{proof}

\begin{theorem}[Sparse recovery, e.g., \cite{FrahlingIS08,JayaramWZ22}]
    \label{thm:sparse-recovery}
    Let $\delta \in (0,1)$ and $k \in \N$ be parameters, and let $\Omega$ be a finite universe. There is a data structure that implicitly maintains a vector $v = (v_x)_{x \in \Omega}$ using $O(k \polylog (|\Omega|/\delta))$ words of space. The vector is initially zero and the data structure supports two operations in time $O(\polylog (|\Omega|/\delta))$:
    \begin{itemize}
        \item \update{$x, b$}: given $x \in \Omega$ and $b \in \set{-1,1}$, updates $v_x := v_x + b$.
        \item \support{}: returns the support size $s := \#\set{x \in \Omega: v_x \neq 0}$ if $s \leq k$, and $\infty$ if $s > k$. The answer is correct with probability at least $1-\delta$. 
    \end{itemize}
\end{theorem}

\subsection{Counting via weak sampling}
We need a variant of $l$-samples in the discrete setting:
\begin{definition}[weak $q$-sample]
    Let $\discrete{U}$ be a finite set and $q \in [0,1]$. A random subset $S \subseteq \discrete{U}$ is a \emph{weak $q$-sample} of $\discrete{U}$ if $\Pr(x \in S) = q$ for all $x \in \discrete{U}$, and the events $\set{x \in S}, \set{y \in S}$ are independent for all pairs of distinct $x,y \in \discrete{U}$.
\end{definition}

\begin{lemma}
    \label{lem:weak-sample}
    If $S$ is a weak $q$-sample of $\discrete{U}$ with $q \geq \frac{20}{\eps^2 |\discrete{U}|}$, then $|S|/q \in (1 \pm \eps) |\discrete{U}|$ with probability at least $19/20$.
\end{lemma}

\begin{proof}
    We bound the expectation and variance of $|S|$. For each $x \in \discrete{U}$ we define a random variable $I_x$ indicating the event $\{x \in S\}$. Clearly $|S| = \sum_{x \in \discrete{U}} I_x$, so $\E |S| = q|\discrete{U}|$. Since the indicator variables are pairwise independent, $\Var |S| = \sum_{x \in \discrete{U}} \Var I_x = q(1-q) |\discrete{U}| \leq \E |S|$. Applying Chebyshev's inequality,
    \[
        \Pr \left( \card{\frac{|S|}{q} - |\discrete{U}|} \geq \eps |\discrete{U}| \right)
        \leq \frac{\Var |S|}{(\eps \E |S|)^2}
        \leq \frac{1}{\eps^2 \E |S|}
        = \frac{1}{\eps^2 q|\discrete{U}|}
        \leq \frac{1}{20}. \qedhere
    \]
    \aftermath
\end{proof}

Let us explain the idea of our data structure. Fix a prime $p \in [\Delta^d, 2 \Delta^d]$ and sample $a,b \in \F_p$ uniformly at random. Define a hash function $h_{a,b}: [\Delta]^d \to \F_p$ by
\[ h_{a,b}(x_1, \dots, x_d) := \left( a + b\,\sum_{i=1}^d \Delta^{i-1} x_i \right) \bmod p. \]
The data structure morally maintains $S^{(l)} := \{x \in \discrete{U} : h_{a,b}(x) \leq \lceil p/2^l \rceil-1\}$ for each $l \in \N$, where $U$ is the union of the surviving objects. Note that $\Pr(x \in S^{(l)}) = \lceil p/2^l\rceil / p =: q_l$, and the events $\{x \in S^{(l)}\}, \{y \in S^{(l)}\}$ are pairwise independent for $x \neq y$. So $S^{(l)}$ is a weak $q_l$-sample of $\discrete{U}$. We then try to apply Lemma~\ref{lem:weak-sample} to retrieve an estimate of $|\discrete{U}|$.

\subsection{Weak sampling from a single object}
\begin{lemma}
    \label{lem:single-weak-sample}
    There is an algorithm \filter{$X, a, b, l$} that, given a convex body $X$, random numbers $a,b \in \F_p$, and an integer $l$, computes $\discrete{X} \cap S^{(l)} = \{x \in \discrete{X} : h_{a,b}(x) \leq \lceil p/2^l \rceil - 1\}$. It runs in time $O(\polylog(\Delta) \cdot |\discrete{X}| / 2^l)$ in expectation over the randomness of $a,b$.
\end{lemma}

\begin{proof}
    We describe the algorithm \filter{$X, a, b, l$} as follows.
    \begin{enumerate}
        \item Compute a bounding box $B \supseteq X$ with $\vol(B) \leq (12d^4)^d \vol(X)$ by Corollary \ref{cor:bounding-box}. Let $H_1, \dots, H_{2d}$ be the defining hyperplanes such that $B = \set{x \in \R^d : H_i(x) \geq 0, \forall i \in [2d]}$.
        
        \item Compute the set of solutions $\class{S}$ to the following integer linear program by Corollary~\ref{cor:ILP}.
        \begin{align*}
            \forall i \in [2d], \quad H_i(x) &\geq 0 \quad \\
            a + b\,\sum_{i=1}^d \Delta^{i-1} x_i &= y + zp \\
            x &\in [\Delta]^d \\
            y &\in \{0, \dots, \lceil p/2^l \rceil - 1\} \\
            z &\in \{0, \dots, d\Delta^d\}
        \end{align*}
        Here, additions and multiplications are over $\Z$ instead of $\F_p$.
        
        \item Return $\set{x : (x,y,z) \in \class{S} \text{ and } X.\contains(x)}$.
    \end{enumerate}
    
    The correctness proof is easy. For $x \in [\Delta]^d$, we write $\eta_{a,b}(x) := a + b \sum_{i=1}^d \Delta^{i-1} x_i$ and note that $0 \leq \eta_{a,b}(x) < p + pd\Delta^d$. Hence, $h_{a,b}(x) = y$ if and only if $\eta_{a,b}(x) = y+zp$ for some (unique) $z \in \set{0, \dots, d\Delta^d}$. So there is a bijection between $\set{x \in \discrete{B} : h_{a,b}(x) \leq \lceil p/2^l \rceil - 1}$ and $\class{S}$. Since $\discrete{B} \supseteq \discrete{X}$, the algorithm returns $\set{x \in \discrete{X} : h_{a,b}(x) \leq \lceil p/2^l \rceil - 1}$.
    
    We next analyze the running time. Step 1 runs in $O(\log n)$ time by Corollary \ref{cor:bounding-box}. For step 2, the ILP contains $d+2$ variables and $2d+1$ linear inequalities. So for constant~$d$ it runs in time $(\log \Delta)^{c_{d+2}'} = \polylog(\Delta)$ per solution by Corollary~\ref{cor:ILP}. Observe that $\E |\class{S}| = q_l |\discrete{B}| = O(|\discrete{B}|/2^l)$ over the randomness of $a,b$. So the expected total time is $O(\polylog(\Delta) |\discrete{B}| / 2^l)$. Recall that $\vol(B) \leq O(\vol(X))$, and that $\frac{|\discrete{B}|}{\vol(B)}, \frac{|\discrete{X}|}{\vol(X)} \in [\frac{1}{2}, 2]$ by Lemma \ref{lem:grid}, so we have $|\discrete{B}| \leq O(|\discrete{X}|)$. Therefore, the running time is $O(\polylog(\Delta) \cdot |\discrete{X}|/2^l)$.
\end{proof}

\subsection{Weak sampling from the union}

\begin{algorithm}[htb]
    \caption{Union volume estimation for dynamic convex bodies}
    \label{alg:convex}
    \Parameter{$n, \eps, \Delta$}
    \Function{\initialize{}}{
        fix a prime $p \in [\Delta^d, 2\Delta^d]$\;
        sample $a, b \in \F_p$ uniformly at random\;
        \For{$l = 0, \dots, \lfloor \log p \rfloor$}{
            instantiate a data structure $v^l$ in Theorem~\ref{thm:sparse-recovery} with $k := \frac{100}{\eps^2}$ and $\delta := \frac{1}{20 \log p}$\;
            $q_l := \lceil p / 2^l \rceil / p$\;
        }
    }

    \Function{\ins{$X$}}{
        \For{$l = 0, \dots, \lfloor \log p \rfloor$}{
            \If{$X.\size \leq 1600 / (\eps^2 q_l)$}{
                \ForEach{$x \in$ \filter{$X, a, b, l$}}{
                    $v^l$.\update{$x, 1$}\;
                }
            }
        }
    }

    \Function{\del{$X$}}{
        \For{$l = 0, \dots, \lfloor \log p \rfloor$}{
            \If{$X.\size \leq 1600 / (\eps^2 q_l)$}{
                \ForEach{$x \in$ \filter{$X, a, b, l$}}{
                    $v^l$.\update{$x, -1$}\;
                }
            }
        }
    }

    \Function{\est{}}{
        $L := 1 + \max \set{l : v^l.\support{} = \infty}$\;
        \Return $v^L.\support{} / q_L$\;
    }
    \end{algorithm}

\begin{theorem}
    \label{thm:convex-point-count}
    Assume $\eps \in (0, \frac{1}{4}]$. Algorithm~\ref{alg:convex} satisfies the following correctness guarantee. Consider an arbitrary sequence of \ins and \del operations, and let $\class{X}$ be the multiset of surviving objects. With probability at least $0.7$, \est{} returns a $(1 \pm \eps)$-approximation to $|\discrete{U}| = |[\Delta]^d \cap U|$ where $U = \bigcup_{X \in \class{X}} X$. Moreover, the algorithm uses $O(\polylog(\Delta)/\eps^2)$ space and expected time per operation.
\end{theorem}

\begin{proof}
    Let us define a hierarchy of random sets $S^{(0)} \supseteq S^{(1)} \supseteq \cdots \supseteq S^{\lfloor \log p \rfloor}$, where
    \[ S^{(l)} := \set{ x \in \discrete{U} : h_{a,b}(x) \leq \lceil p/2^l \rceil - 1 }. \]
    Note that $S^{(l})$ is a weak $q_l$-sample of $\discrete{U}$, where $q_l := \lceil p/2^l \rceil / p \in [2^{-l}, 2^{-l+1}]$. Since $q_l/q_r \in [2^{r-l-1}, 2^{r-l+1}]$ for all $l,r$, there exists an integer $l^* \in \N$ such that $20/\eps^2 \leq q_{l^*} |\discrete{U}| < 80/\eps^2$.
    
    For $l \geq l^*-4$, we make some observations on the sparse vectors $v^l$ in the algorithm. Every surviving object $X \in \class{X}$ satisfies $X.\size \leq \vol(U) \leq (1+\eps) |\discrete{U}| \leq \frac{5}{4} \cdot 80 / (\eps^2 q_{l^*}) \leq 1600 / (\eps^2 q_l)$, which satisfies the if-condition at $l$ when $X$ was inserted. The if-block thus increments the vector $v^l$ on coordinates $X \cap S^{(l)}$. On the other hand, every non-surviving object is treated consistently during insertion and deletion, so it does not contribute to $v^l$. Therefore, $\supp(v^l) = \bigcup_{X \in \class{X}} (X \cap S^{(l)}) = S^{(l)}$. In particular, $\supp(v^{l^*-4}) \supseteq \cdots \supseteq \supp(v^{\lfloor \log p \rfloor})$.

    Now consider a call to \est{}. Let $F$ be the event that all sparse vectors behave correctly according to Theorem~\ref{thm:sparse-recovery}. By a union bound, $\Pr(F) \geq 19/20$. We also define an event $E_l := \set{|S^{(l)}|/q_l \in (1 \pm \eps) |\discrete{U}|}$ for $l^*-4 \leq l \leq l^*$. Then $\Pr(E_l) \geq 19/20$ by Lemma~\ref{lem:weak-sample}. Under the event $F \cap \bigcap_{l=l^*-4}^{l^*} E_l$, we have
    \begin{itemize}
        \item $\card{S^{l^*-4}} \geq (1-\eps) \cdot q_{l^*-4} |\discrete{U}| \geq \frac{3}{4} \cdot \frac{20}{\eps^2} \frac{q_{l^*-4}}{q_{l^*}} \geq \frac{120}{\eps^2} > k$. So $v^{l^*-4}.\support{} = \infty$.
        \item $\card{S^{l^*}} \leq (1+\eps) \cdot q_{l^*} |\discrete{U}| \leq \frac{5}{4} \cdot \frac{80}{\eps^2} = \frac{100}{\eps^2} = k$. So $v^l.\support{} < \infty$ for all $l \geq l^*$.
    \end{itemize}
    So the algorithm computes $l^* - 3 \leq L \leq l^*$ and outputs $v^L.\support{}/q_L = |S^{(L)}|/q_L$, which is a $(1 \pm \eps)$-approximation to $|\discrete{U}|$. The success probability is thus at least $\Pr( F \cap \bigcap_{l=l^*-4}^{l^*} E_l ) \geq \frac{19}{20} - 5 \cdot \frac{1}{20} = 0.7$.

    Next we analyze the complexity. The space is dominated by $\log p = O(\log \Delta)$ sparse vectors, each using $O(k\polylog(\Delta))$ words by Theorem~\ref{thm:sparse-recovery}. So the space complexity is $O(\polylog(\Delta)/\eps^2)$. Regarding time complexity, each update has $\log p = O(\log \Delta)$ iterations, and each iteration runs in expected time $O(\polylog(\Delta) \cdot 2^{-l} \cdot 2^l/\eps^2) = O(\polylog(\Delta) / \eps^2)$ by Lemma~\ref{lem:single-weak-sample}. So the update time is $O(\polylog(\Delta) / \eps^2)$ as claimed.
\end{proof}

Now Theorem~\ref{thm:convexalgo} follows directly from Theorem~\ref{thm:convex-point-count} and Lemma~\ref{lem:discretization}. (To amplify the success probability, we run $O(\log n)$ independent copies and output the median.)

\bibliography{ref}

\appendix

\section{Fat Objects Are Easy}
\label{sec:fat}

We briefly mention that for fat objects in a constant dimension, there are simple ways to solve the dynamic union volume estimation. To this end, we can build a quadtree and replace each object $X$ by a union $\widetilde{X}$ of $O(1/\eps^d)$ quadtree cells of side length $O(\eps \cdot \mathrm{diam}(X))$ where $\mathrm{diam}(X)$ is the diameter of $X$.  It is straightforward to maintain the volume of the union of all these quadtree cells, using a tree structure augmented with intermediate volumes stored at each node. (We may assume that the aspect ratio is polynomial in $n$ due to Lemma \ref{lem:reduce-aspect-ratio} and so the tree height is $O(\log n)$.)

For fat objects, it is not difficult to show that the volume of $\bigcup_{X\in \class{X}} \widetilde{X}$ (which is contained in a union of a $(1+O(\eps))$-factor expansion of $X$) gives a good approximation to the volume of $\bigcup_{X \in \class{X}} X$.
It may even be possible to adapt this approach in streaming settings by incorporating known streaming techniques.  

The problem becomes more nontrivial and subtle when the objects may not be fat: First, we cannot afford to replace a non-fat object with quadtree cells. Second, and more curiously, for non-fat objects, if we take some $(1+\eps)$-factor expansion $\widetilde{X}$ of each object $X$, it is not obvious how the volume of $\bigcup_{X\in \class{X}}\widetilde{X}$ relates to the volume of $\bigcup_{X\in \class{X}} X$.  In fact, this type of question (even in the special case when the objects are narrow ``rods'') is related to the infamous \emph{Kakeya set conjecture}, for which considerable recent effort has been spent on just the 3D case~\cite{WangZ25} (and much less is known in 4D and beyond)!

\end{document}